\documentclass{article}
\usepackage{authblk}
\usepackage{graphicx}
\usepackage{amsmath}
\usepackage{amssymb}
\usepackage{cmll}
\usepackage{amsthm}
\usepackage{stmaryrd}
\usepackage{MnSymbol}
\usepackage{tikz}
\usetikzlibrary{calc,positioning}

\newcommand{\hide}[1]{}

\newcommand{\St}{\hbox{St}}
\newcommand{\Stop}{\hbox{\rm Stop}}
\newcommand{\tick}{\checkmark}

\newcommand{\B}{{\mathcal B}}

\newcommand{\sncirc}{\oast} 

\newtheorem{theorem}{Theorem}[section]

\newtheorem{prop}[theorem]{Proposition}
\newtheorem{lemma}[theorem]{Lemma}
\newtheorem{definition}[theorem]{Definition}
\newtheorem{notation}[theorem]{Notation}
\newtheorem{example}[theorem]{Example}

\theoremstyle{remark}

\newtheorem*{theorem*}{Theorem}

\def\endex{{\hspace*{\fill}\hbox{$\Box$}}}

 \usepackage[all]{xy}

\newcommand{\opmove}{{\boxminus}}
\newcommand{\plmove}{{\boxplus}}
\newcommand{\nemove}{\footnotesize{\boxcircle}}

\def\profto{\!\!\!\xymatrix@C-.75pc{\ar[r]|-{\! +\!} &}\!\!\! }

\def\all{\forall}

\newdir{pb}{:(1,-1)@^{|-}}
\def\pb#1{\save[]+<16 pt,0 pt>:a(#1)\ar@{pb{}}[]\restore}

\newcommand{\vis}[1]{{{#1}_\downarrow}}

\newcommand{\co}{\mathbin{{\it co}}}
\newcommand{\vvbar}{{\mathbin{\parallel}}}
\newcommand{\scirc}{{{\odot}}}

\newcommand{\imc}{\rightarrowtriangle}
\newcommand{\setdif}{\setminus}

\newcommand{\al}{\alpha}

\newcommand{\ga}{\gamma}
\newcommand{\sig}{\sigma}

\newcommand{\fsubseteq}{\subseteq_{\rm fin}}

\newcommand{\CC}{{\rm C\!\!C}}
\newcommand{\cc}{\,c\!c\,}
\newcommand{\pol}{{\it pol}}

\renewcommand{\mod}[1]{{|{{#1}}|}}
\newcommand{\Strat}{{\mathbf{Strat}}}
\newcommand{\NStrat}{{\mathbf{BStrat}}}
\newcommand{\SStrat}{{\mathbf{SStrat}}}

\newcommand{\aug}{{{\it aug}}}

  \newcommand{\spanpl}[5]{{\xymatrix{
    & {#3}\ar@{_{(}->}[dl]_{#2}\ar[dr]^-{#4} &\\
      {#1} && {#5}
  }}}

\def\Con{{\rm Con}}

\newcommand{\eswp}{event structure with polarity}
\newcommand{\esswp}{event structures with polarity}

\newcommand{\conf}[1]{\:\!{\cal C}(#1)}
\newcommand{\iconf}[1]{\:\!{\cal C}^\infty(#1)}

\newcommand{\parrow}{\rightharpoonup}

\newcommand{\arr}[1]{{{\stackrel{#1}{\longrightarrow}}}}
\newcommand{\id}{{\rm id}}

\newcommand{\arrow}{\rightarrow}
\newcommand{\set}[2]{{\{  #1\  | \  #2 \} }}
\newcommand{\setof}[1]{{\{ #1 \} }}

\newcommand{\eqdef}{\mathrel{=_{\mathrm{def}}}}

\newcommand{\iso}{\cong}

\def\mxxth{\mathsurround=0pt}
\dimendef\dimenxx=0
\def\openup{\afterassignment\xxpenup\dimenxx=}
\def\xxpenup{\advance\lineskip\dimenxx
  \advance\baselineskip\dimenxx \advance\lineskiplimit\dimenxx}
\def\eqalign#1{\,\vcenter{\openup1\jot \mxxth
  \ialign{\strut\hfil$\displaystyle{##}$&$\displaystyle{{}##}$\hfil
     \crcr#1\crcr}}\,}

\newif\ifdtxxp
\def\displxxy{\global\dtxxptrue \openup1\jot \mxxth
  \everycr{\noalign{\ifdtxxp \global\dtxxpfalse
      \vskip-\lineskiplimit \vskip\normallineskiplimit
      \else \penalty\interdisplaylinepenalty \fi}}}
\def\displaylines#1{\displxxy
  \halign{\hbox to\displaywidth{$\hfil\displaystyle##\hfil$}\crcr
      #1\crcr}}

\newskip\mycntring \mycntring=0pt plus 1000pt minus 1000pt

\def\leqalignno#1{\displxxy \tabskip=\mycntring
  \halign to\displaywidth{\hfil$\displaystyle{##}$\tabskip=0pt
      &$\displaystyle{{}##}$\hfil\tabskip=\mycntring
      &\kern-\displaywidth\rlap{$##$}\tabskip=\displaywidth\crcr
      #1\crcr}}

\newcommand{\ER}{{\cal E}_r}
\newcommand{\ET}{{\cal E}_t}

\newcommand{\ie}{{\it i.e.}}
\newcommand{\eg}{{\it e.g.}}

\newdir{C}{%
!/4.5pt/@{( }*:(1,-.2)@{}*:(1,+.2)@_{}}

\newdir{|>}{%
!/4.5pt/@{|}*:(1,-.2)@{>}*:(1,+.2)@_{>}}

\begin{document}

\title{\bf The Mays and Musts of Concurrent Strategies}
\author[1]{\rm Simon Castellan}
\author[2]{\rm Pierre Clairambault}
\author[3]{\rm Glynn Winskel}
\affil[1]{Inria, Univ Rennes, IRISA, France}
\affil[2]{
Univ Lyon, EnsL, UCBL, CNRS,  LIP, F-69342, LYON Cedex 07, France}
\affil[3]{Computer and Information Sciences, University of Strathclyde, Scotland}
\date{}
\maketitle

\begin{abstract}{Concurrent strategies based on event structures are examined from the viewpoint of `may' and `must'  testing in traditional process calculi. In their pure form concurrent strategies fail to expose the deadlocks and divergences that can arise in their composition.  This motivates an extension of the bicategory of 
concurrent strategies to treat the `may' and `must' behaviour of strategies under testing. 
One extension adjoins neutral moves to strategies but in so doing loses identities w.r.t.~composition. This in turn  motivates another extension in which concurrent strategies are accompanied by stopping configurations; the ensuing stopping strategies inherit the structure of a bicategory from that of strategies.  The technical developments converge in providing characterisations of the `may' and `must' equivalences and preorders on strategies. }\end{abstract}

\section{Introduction}
This article relates to work on process calculi of the 1980's but from a modern perspective of processes as strategies, specifically as distributed/concurrent strategies based on event structures. 
It expands on two areas close to Samson Abramsky's heart, game semantics and concurrency: on a development of concurrent games based on event structures which extends his early ideas with Paul-Andr\'e Melli\`es of deterministic concurrent strategies as closure operators
~\cite{Congames}; and equivalences on concurrent processes through testing~\cite{abramskytesting}.  

Robin Milner and Tony Hoare's work of late seventies and early eighties 
drew attention to equivalences on processes; Milner's on forms of {\em bisimulation}~\cite{milner} and Hoare's on {\em failures} equivalence~\cite{hoare}.  
Hoare had described failure equivalence informally as the minimum extension of trace equivalence that takes account of the possibility of failure due to deadlock. 
Matthew Hennessy and his PhD student Rocco de Nicola provided a rationale 
 through an idea of testing processes~\cite{hennessy-denicola}.  For them a test was a process with distinguished ``success'' states at which an action $\tick$ could occur.  Putting a test in parallel composition with a process, {\em may} lead to success if some run does  or {\em must} lead to success if all runs do. Processes can be regarded as equivalent if they  have the same `may' and `must' behaviour  w.r.t.~tests. Modulo subtleties to do with the divergence of processes, Hennessy and de Nicola recovered failure equivalence as testing equivalence. What about Milner's central equivalence? Samson Abramsky investigated the extent to which bisimulation could be viewed as a testing equivalence~\cite{abramskytesting}: it could, but only at the cost of strengthening the power of tests considerably, by allowing testing to run and copy processes quite liberally.  

Here we shall examine the `may' and `must' equivalence of concurrent strategies based on event structures~\cite{lics11,LMCS}---foreshadowed in the early definitions of concurrent strategy~\cite{Congames,murawski,Asgames,FP}.  
Informally, a 
{\em strategy} for Player in a two-party game against Opponent, expresses a choice of Player moves, most often in reaction to moves made by Opponent, 
unpredictable for Player but for the constraints of the game. 
We shall implicitly regard a strategy as a strategy for Player. We regard Opponent as the environment uncontrollable by Player.
We can express both the game---its moves and their constraints---and a strategy---its choice of Player moves subject to the moves of Opponent---as event structures.  This chimes with our view of strategies and games as highly distributed.    Player and Opponent are more accurately thought of as teams of players and opponents acting at possibly very different locations.  Though we take the rather abstract view of location advocated by Petri in his concept of local state as a condition (or place): then locality reveals itself through the causal dependence and independence of events.  

Event structures are the concurrent analogue of trees; just as transition systems unfold to trees, so Petri nets unfold to event structures.  Whereas an unfolded behaviour of a transition system comprises  sequences of actions/events, the unfolded behaviour of a Petri net, in which events make local changes to conditions,
comprises partial orders of causal dependency between event occurrences~\cite{NPW}.  Event structures are a central model for concurrent computation, related to other models by adjunctions~\cite{WN}. 
This plants concurrent strategies based on event structures firmly within theories of concurrency and interaction---anticipated in Abramsky's presentation of game semantics, with its emphasis on composition of strategies as given by their parallel interaction followed by hiding.
Perhaps more controversially, the view of processes as strategies suggests refinements to the assumptions usual in process calculi.   In concurrent strategies, gone is the usual symmetry between a process and its environment; the conditions on a concurrent strategy take account of the unpredictability and uncontrollability of Opponent moves. This affects the appropriate equivalences to impose between concurrent strategies.  
 
There is surely a long history behind the idea of composing strategies.  Certainly the idea plays a key role in John Conway's  ``On Numbers and Games"~\cite{conway}, the categorical underpinnings of which were exposed by Andr\'e Joyal~\cite{joyal}.  For two-party games there is the obvious operation of reversing the roles of the two participants, Player and Opponent; this operation, forming the {\em dual} $G^\perp$ of a game $G$, played the role of {\em negation} for Conway. A useful convention is to regard a strategy in a game $G$ as a strategy for Player; then a strategy for Opponent, or counter-strategy, is a strategy in the dual game $G^\perp$.   
If the games are broad enough, they often support a form of parallel composition, $G\vvbar H$; for Conway it was the {\em sum} of games.  A strategy {\em from} a game $G$ {\em to} a game $H$ is a strategy $\sig$  {\em in} the game $G^\perp\vvbar H$.  Given another strategy this time from the game $H$ to the game $K$, \ie~a strategy $\tau$ in the game $H^\perp\vvbar K$, we can let the strategies {\em interact} as $\tau\sncirc \sig$, essentially by playing them against each other over the common game $H$; there the strategies $\sig$  and $\tau$ adopt complementary roles---where one makes a move of Player in $H$ the other sees a move of Opponent and {\it vice versa}.  

The interaction $\tau\sncirc \sig$ involves moves in the parallel composition of all three games, $G^\perp\vvbar H\vvbar K$, though in writing the parallel composition in this way an imprecision has crept in: whereas the moves over $G^\perp$ and $K$ described by $\tau\sncirc \sig$ are choices of moves for Player or moves open to choices of Opponent, those over $H$ 
are either instantiations of Opponent moves of $\sig$ by Player moves of $\tau$, or the converse, instantiations of Opponent moves of $\tau$ by Player moves of $\sig$.  As such the moves of $\tau\sncirc \sig$ over $H$ behave like synchronisations between complementary moves of $\sig$ and $\tau$,  and as events internal to the interaction. Though internal, the events over $H$ can affect the behaviour of the interaction by introducing deadlocks or divergence.  In the {\em composition} of strategies it is usual to hide the internal events of interaction to obtain a strategy $\tau\scirc \tau$ in the game $G^\perp\vvbar K$, where the game $H$ is  elided to obtain a strategy from $G$ to $K$.  However when the original strategies $\sig$ or $\tau$ are nondeterministic significant behavioural distinctions can be lost in hiding internal events.  In particular, the hidden events can affect the `must' behaviour of the composition of strategies.

\subsection{Contributions of the paper}
This brings us to the concerns of this paper.
It motivates the definition of {\em bare} concurrent strategies in which internal events are exposed as neutral moves in  the strategy.  Through bare strategies we can examine the `may' and `must' testing of strategies.
Although bare strategies compose their composition does not have identities, so they fail to form a bicategory.  We have explored two ways to recover a bicategory while remaining faithful to the `must' behaviour of strategies.  One is through ``essential events'' in which one strips a bare strategy down to just those neutral moves critical to its behaviour~\cite{essential}.  The other, that we follow here, is through extending strategies with the extra structure of 
{\em stopping configurations}~\cite{stratsproc}.  By distinguishing certain configurations as stopping we keep track of those visible configurations at which the strategy may appear to get stuck through the occurrence of hidden neutral moves. 

Stopping configurations are the event-structure analogue of  
Russ Harmer and Guy McCusker's ``divergences''~\cite{HarmerM99}, though 
event structures add the refinement of locality and independence to the concept.  In an interleaving model, in which behaviour is captured through sequences of actions, divergence anywhere has a global effect;  generally, in the parallel composition of two processes if one can perform an infinite sequence of actions, these may block progress of the other process,  unless additional fairness assumptions are enforced.
This is not so in a model such as event structures where the independence/concurrency of actions is explicit.  
In a nondeterministic strategy  
a Player move  
is not blocked by the occurrence of moves with which it is independent. Whereas an interleaving model  may require weak fairness assumptions these are generally built into the behaviour of strategies as event structures~\cite{thesis}.  This makes for subtle differences in the nature of `must'  testing in concurrent games w.r.t.~traditional games.

 Perhaps surprisingly, in many situations a concurrent strategy may be replaced by its simpler ``rigid image'' in the game, despite this often forgetting nondeterministic branching---rigid-image strategies form a category rather than just a bicategory~\cite{madeeasy}---and indeed this remains true for strategies with stopping configurations; none of the `may' or `must' behaviour is lost.

Our technical contribution concludes with characterisations of the `may' and `must' equivalences and preorders on strategies.  
The `may' equivalence of strategies is captured through their inducing the same set of finite traces; a trace being understood as a sequence of moves in the game.  This echoes the earlier results of  Ghica and Murawski when showing  their non-alternating
games model is fully-abstract for Idealized Parallel Algol with
respect to may-convergence~\cite{murawski}.
 For ‘must’ equivalence, our result is to be compared
with that of Harmer and McCusker for their sequential games model, based on Hyland-Ong
games with explicit divergences~\cite{HarmerM99}.  
But whereas
in a sequential setting only the first divergence matters
, for us 
`must' equivalence of strategies is equivalent to their sharing the same  traces of {\em all} (possibly infinite) stopping configurations.  See
Example~\ref{ex:musttesting} and what follows for an in-depth discussion.

Because we restrict attention to linear bicategories of strategies, the tests here are linear too; they do not permit the tested strategy to be copied and rerun.  From the point of view of distributed computation linearity is natural:  it is often infeasible to copy a distributed system or strategy~\cite{winskel1999}.   Then single-run   `may' and  `must'  tests are appropriate.

On the other hand, most programming languages do allow some form of copying and nonlinearity. 
Through the addition of symmetry we can adjoin pseudo (co)monads---where the traditional laws hold up to symmetry, and model nonlinear features~\cite{lics14, coHO}.  
Through symmetry and pseudo comonads, we can realise a variety of nonlinear forms of testing, in which a test could dynamically copy and retest the strategy of interest.  The nature of such broader testing on strategies, the equivalences and logics  induced, are not well understood and deserve a systematic study, for which  this paper forms a foundation.   As shown by Samson Abramsky such broader tests are needed to realise
equivalences such as bisimulation as testing equivalences~\cite{abramskytesting}.

\section{Event structures}
An {\em  event structure}  comprises $(E, \leq, \Con)$, consisting of a set $E$ of {\em events}  
which are
partially ordered by $\leq$, the {\em causal dependency
relation},
and a  nonempty {\em consistency} relation $\Con$ consisting of finite subsets of $E$.  The relation
$e'\leq e$ expresses that event $e$ causally depends on the previous occurrence of event $e'$.  That a finite subset of events is consistent conveys that its events can occur together by some stage in the evolution of the process.  
Together the relations satisfy several  axioms.  We insist that the partial order is {\em finitary}, \ie
\begin{itemize}
\item $[e]\eqdef \set{e'}{e'\leq e}\hbox{ is finite for all } e\in E$\,,
\end{itemize}
and that consistency satisfies 
\begin{itemize}
\item 
$\setof{e}\in\Con \hbox{  for all } e\in E$\,,
\item 
$Y\subseteq X\in\Con \hbox{ implies }Y\in \Con,\ \hbox{ and}$
\item
$
X\in\Con \ \&\  e\leq e'\in X \hbox{ implies } 
X\cup\setof{e}\in\Con$\,.
\end{itemize}
There is an accompanying notion of state, or history, those events that may occur  up to some stage in the behaviour of the process described.  A {\em configuration} is a, possibly infinite, set of events $x\subseteq E$ which is:  
\begin{itemize} 
\item
{\em consistent,} $X\subseteq x  \hbox{ and }  X \hbox{ is finite}  \hbox{ implies } X\in\Con$\,; 
and
\item
{\em down-closed, }
$ e'\leq e\in x  \hbox{ implies } e' \in x$\,.
 \end{itemize}

Two events $e, e'$ are  called {\em concurrent} if  the set $\setof{e,e'}$ is in $\Con$ and neither event is causally dependent on the other; then we write $e\co e'$.  In games the relation of {\em immediate} dependency $e\imc e'$, meaning $e$ and $e'$ are distinct with $e\leq e'$ and no event in between,  plays a very important role.
  We write $[X]$ for the down-closure of a subset of events $X$.  
Write $\iconf E$ for the configurations of $E$ and $\conf E$ for its finite configurations.   (Sometimes we shall need to distinguish the precise event structure to which a relation is associated and write, for instance, $\leq_E$,  $\imc_E$ or $\co_E$.)

\begin{example}{\rm 
In examples it is often convenient to draw event structures.  Often, though not always, consistency is determined in a binary fashion, in that a set of events is consistent if all of its pairs are.
 Then, it is economical to draw  the binary relation of {\em conflict}, or inconsistency.  For example, in the diagram
$$
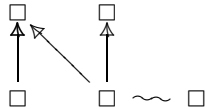
\xymatrix@R=20pt@C=20pt{
\boxempty&\boxempty&\\
\boxempty\ar@{|>}[u]\ar@{|>}[u]&\boxempty\ar@{|>}[ul]\ar@{|>}[u]\ar@{~}[r]& \boxempty
}$$ 
we illustrate the relations of immediate causal dependency $\xymatrix{\ar@{|>}[r]&}$ which yields the Hasse diagram of the partial order of causal dependency between  events $\boxempty$, and conflict  by the wiggly line $\xymatrix{\ar@{~}[r]&}$.
Neither the two events related by $\xymatrix{\ar@{~}[r]&}$ nor their dependants w.r.t.~causal dependency can occur together in a configuration; there is no need draw all the conflicts that follow.
}\endex\end{example}

Let  $E$ and $E'$ be event structures.
A {\em  map} of event structures 
$f:E\arrow E'$ 
 is a partial function on events
$f:E\parrow E'$ such that 
for all   $x\in\iconf E$
its direct image $f x\in\iconf{E'}$  and 
$$
\hbox{if } e_1, e_2 \in x 
\hbox{ and } f(e_1) =f(e_2) \hbox{ (with both defined)}, \hbox{ then } e_1=e_2.
$$
(Those maps defined is unaffected if we replace possibly infinite configurations $\iconf E$ by finite configurations $\conf E$  above; this is because any configuration is the union of finite configurations and direct image preserves such unions.)

Maps of event structures compose as partial functions, with identity maps given by identity functions.  
Say a map is {\em total} if the function $f$ is total.  Notice that  for a total map $f$ the condition on maps now says it is {\em  locally injective}, in the sense that w.r.t.~any configuration $x$ of the domain the restriction of $f$ to a function from $x$   is injective; the restriction of $f$ to a function from $x$ to $f x$ is thus bijective.  Say a total map of event structures is {\em rigid} when it preserves causal dependency. 

Although a map $f:E\arrow E'$ of event structures does not generally preserve causal dependency, it does locally reflect   
causal dependency:  whenever $e, e'\in x$, a configuration of $E$, and $f(e)$ and $f(e')$ are both defined with $f(e')\leq f(e)$, then $e'\leq e$.  Consequently,
$f$  preserves the concurrency relation: if $e\co e'$ in $E$  
and $f(e)$ and $f(e')$ are both defined then $f(e)\co f(e')$.

\section{Constructions}

We provide the constructions which we use in the paper.

\subsection{Partial-total factorisation}

We shall realise an operation of hiding events via a factorisation property of maps of event structures.  

Let $(E,\leq, \Con)$ be an event structure.  Let $V\subseteq E$ be a subset of `visible' events.
Define  
$
E{\mathbin\downarrow} V\eqdef
(V, \leq_V, \Con_V)
$, 
where
$v \leq_V v' \hbox{ iff } v\leq v' \ \&\ v,v'\in V$ and $X\in\Con_V \hbox{ iff }  X\in\Con\ \&\ X\subseteq V$. The operation projects $E$ to visible events $V$.

Consider a partial map of event structures $f:E\to E'$.  Let 
$$V\eqdef \set{e\in E}{ f(e) \hbox{ is defined}}\,.$$
Then $f$ clearly factors into the composition 
$$
\xymatrix{
E\ar[r]^{f_0}& E{\mathbin\downarrow} V \ar[r]^{f_1}& E'\\}
$$
of $f_0$, a partial map of event structures taking $e\in E$ to itself if $e\in V$ and undefined otherwise, and $f_1$, a total map of event structures acting like $f$ on $V$. We call $f_1$ the {\em defined part} of the partial map $f$.   We say a map $f:E \to E'$ is a {\em projection} if its defined part is an isomorphism.  

The {\em partial-total factorisation} is characterised to within isomorphism by the following 
  universal property:   for any factorisation 
$$
\xymatrix{
f:E\ar[r]^{g_0}& E_1 \ar[r]^{g_1}& E' }
$$ 
where $g_0$ is partial and $g_1$ is total there is a  (necessarily total) unique map $h: E{\mathbin\downarrow} V\to E_1$  such that 
$$
\xymatrix
{
E\ar[r]^{f_0}\ar[dr]_{g_0}& E{\mathbin\downarrow} V\ar@{-->}[d]^h \ar[r]^{f_1}& E' \\
 & E_1\ar[ur]_{g_1}& 
}
$$
commutes.


\subsection{Pullback}

Event structures and their maps have pullbacks.  For the composition of strategies we shall only need pullbacks of total maps.
Consider a pullback 
$$
\xymatrix{
&  &\ar[ld]_{\pi_1}P\pb{270}\ar[rd]^{\pi_2}&&\\
 &A\ar[rd]_{f}&&B\ar[ld]^{g}&\\
&&C&&}
$$ 
where $f$ and $g$ are total. Pullbacks are difficult to construct directly on  the ``prime'' event structures we are using here, essentially because they associate each event with a unique minimum causal history.  Such constructs are best first carried out in a broader model.  Here we build the pullback of event structures out of the stable family  of secured bijections.  
 
\begin{definition}{\rm 
 A {\em secured bijection} comprises a composite bijection $$\theta_{x,y}: x\iso fx = gy\iso y$$ between configurations $x\in\iconf A$ and $y\in\iconf B$ s.t.~$f x= g y$, which is {\em secured} in the sense that  the transitive relation generated on $\theta_{x,y}$ by taking  
\begin{center}
$(a,b)\leq (a',b')$ if $a\leq_A a'$ or $b\leq_B b'$
\end{center}
is a finitary partial order. Let $\cal B$ be the family of secured bijections.  Say a subset $Z\subseteq \B$ is {\em compatible} iff  $\exists \theta' \in {\cal B} \forall \theta\in Z.\  \theta \subseteq \theta'$. 
}\end {definition}

\begin{prop}
The family $\cal B$ is a {\em stable family},\footnote{Here it is useful to allow stable families to have infinite configurations, as  originally~\cite{icalp82,evstrs}.}

 \ie~it is
\begin{itemize}
\item
  {\em  Complete:}    
 $\all Z\subseteq\B.  \   Z\hbox{ is compatible } \implies \bigcup Z \in \B$\,;
\item
{\em  Stable:}
$
\all Z\subseteq\B.\ Z\not= \emptyset  \ \& \    Z\hbox{ is compatible } 
\implies 
\bigcap Z\in \B$;
\item
{\em  Finitary:}
$
\all \theta\in\B, (a,b)\in \theta\exists \theta_0\in\B.\ \theta_0 \hbox{ is finite }\ \&\ (a,b)\in \theta_0\subseteq\theta$; and 
\item{\em  Coincidence-free:}   For all
$\theta \in \B$,  $ (a,b), (a',b')
\in \theta$ with $(a,b) \not=  (a',b')$,
$$
 \exists \theta_0 \in \B.\   
 \theta_0  \subseteq  \theta  \ \& \  ((a,b) \in  \theta_0  \iff (a',b') \notin  \theta_0)\, .$$ 
\end{itemize}
\end{prop}

We now apply a general construction $\Pr(\B)$  for obtaining an event structure from the stable family $\B$.  
Suppose  $(a,b)\in \theta$ where $\theta\in\B$~\cite{icalp82,evstrs}.  Because $\B$ is a stable family
$$
[ (a,b) ]_\theta \eqdef \bigcap \set{\phi\in \B}{\phi\subseteq \theta \ \& \   (a,b)\in \phi}\in\B
$$ 
and moreover is a finite set; it represents a minimal way in which $(a,b)$ can occur.  We build the pullback of event structures taking such minimal elements as events.

\begin{prop}\label{lem:pbcharn} 
Defining $\Pr(\B) =
(P, \Con, \leq)$  where:
$$
\eqalign{
&P= \set{[ (a,b) ]_\theta }{(a,b)\in \theta\ \&\ \theta \in \B}\ ,\cr
&
Z \in \Con  \hbox{ iff }  Z \subseteq P \ \&\ \bigcup Z \in\B\  \hbox{ and}\cr
&
p\leq p' \hbox{ iff }  p, p'\in P\ \&\ p\subseteq p'\,}
$$
yields an event structure. 
There is an order isomorphism $$\beta: (\conf{\Pr(\B)},\subseteq) \iso (\B, \subseteq)$$
where $\beta(y) = 
\bigcup y$ for $y\in \conf{\Pr({\B})}$
; its  mutual inverse is $\ga$ where $\ga(\theta) = \set{[(a,b)]_\theta}{(a,b)\in \theta}$ for  $\theta\in\B$.
\end{prop}

There are obvious maps $\pi_1:\Pr(\B)\to A$ and $\pi_2: \Pr(\B)\to B$ given by $\pi_1([ (a,b) ]_\theta) =a$ and $\pi_2([ (a,b) ]_\theta) =b$.  These make the required pullback $\Pr(\B)$, $\pi_1$, $\pi_2$ of event structures.  Why?  
The family $\B$ is a pullback in the category of stable families (its maps are similar to those of event structures).  There is a coreflection from the category of event structures to that of stable families.  Its right adjoint is $\Pr$ which consequently preserves pullbacks, yielding the pullback of event structures when applied to $\B$~\cite{icalp82,ecsym-notes}. 
 

\begin{definition}{\rm 
We shall 
write $x\wedge y$ for the configuration $\ga(\theta_{x,y})$ of $\Pr(\B)$ which correponds to a secured bijection $\theta_{x,y}: x\iso fx = gy\iso y$ between $x\in\iconf A$ and $y\in\iconf B$.   Note that any configuration of the pullback is of the form $x\wedge y$ for unique $x\in\iconf A$ and $y\in\iconf B$. Of course, given $x\in\iconf A$ and $y\in\iconf B$ we cannot be assured that they form a secured bijection even when $fx= gy$. We shall treat $\wedge$ as a partial operation with  $x\wedge y$ only defined when $x\in\iconf A$ and $y\in\iconf B$ form a secured bijection.  
}\end{definition}

\section{Rigid image}\label{sec:rigidimage}

This section is only used late on in the paper when showing how `may' and `must' behaviour transfer to the rigid image of a strategy---Section~\ref{sec:axioms}.   

There is an adjunction  between $\ER$, the category of event structures with rigid maps, to $\ET$, the category of event structures with total maps. Its right adjoint's action on  an event structure $B$ is given as follows. 
For $x \in\iconf B $, an {\em augmentation} of $x$ is a finitary partial order 
$(x, \alpha)$ where $\forall b, b'\in x.\  b\leq_B b' \implies b\,\alpha\, b'$.
We can regard such augmentations as elementary event structures in which all subsets of events are consistent.  
Order  all augmentations by taking
$(x,\alpha) \hookrightarrow (x',\alpha') $
iff
$x\subseteq x'$ and 
the inclusion $i:x\hookrightarrow x'$  is a rigid map  $i: (x,\alpha) \to (x',\alpha')$.  Augmentations under $ \hookrightarrow$  form a prime algebraic domain~\cite{NPW,PrimeAlgDom}, so are isomorphic to the configurations
of an event structure, $\aug(B)$
; its events are the complete primes, which are precisely the augmentations with a top element.

\begin{prop} \label{propn:rigtotadjn}\cite{ESS} The inclusion functor $\ER\hookrightarrow\ET$  
has a right adjoint $\aug$.  
The category $\ET$ is isomorphic to the Kleisli category of the monad induced on $\ER$ 
by the adjunction.
\end{prop}

Rigid maps $f:A\to B$ have a useful image given by restricting the causal dependency of $B$ to  the set of events $f A$, the direct image of the events of $A$, and taking a finite set of events to be consistent if they are the image of a consistent set in $A$.
More generally, a total map $f:A\to B$ 
has a {\em rigid image}  given by the image of its corresponding Kleisli map, the rigid map $\bar f:A\to \aug(B)$.
 Put more directly, a total map  $f:A\to B$ has a {\em rigid image} comprising a factorisation $f= f_1 f_0$ where
$f_0$ is rigid epi and $f_1$ is a total map,
$$\xymatrix{
A\ar[dr]_f\ar@{>>}[r]^{f_0} & B_0\ar[d]^{f_1}\\ 
&B\,,}
$$
 with the following universal property: for any factorisation of $f= f_1' f_0'$ where $f_0'$ is rigid epi, there is a unique map $h$ such that the diagram
$$\xymatrix{
A\ar@{>>}@/^1.5pc/[rr]^{f_0}
\ar[dr]_f\ar@{>>}[r]^{f'_0} & B_0'\ar[d]|{f_1'}\ar@{-->}[r]^h&B_0\ar[dl]^{f_1}\\ 
&B&}
$$
commutes; the map $h$ is necessarily also rigid and epi.   

From the universal property of rigid image we derive:

\begin{prop}
Let $f:A\to B$ and $g:B\to C$  be maps of event structures.  Assume that $f$  is   rigid and epi.   Then, 
$g$ and  $g\circ f$ have the same rigid image.
\end{prop}

\section{Event structures with polarity}

Both games and strategies will be represented by \esswp.
An {\em \eswp}~comprises $(A, \pol)$ where $A$ is an event structure with a polarity function $\pol_A:A\to\setof{+,-, 0}$ ascribing a 
 polarity $+$ (Player), $-$ (Opponent) or $0$ (neutral)  to its events. The events correspond to (occurrences of) moves. 
It will be  technically useful to allow events of neutral polarity; they arise, for example, in the interaction between a strategy and a counterstrategy. We write $A^0$ for the \eswp~in which all the polarities are reassigned $0$, so made neutral. 
A {\em game} shall be represented by an \eswp~in which no moves are neutral. 

\begin{notation}{\rm 
In an \eswp~$(A,\pol)$, with configurations $x$ and $y$, write 
$x\subseteq^-y$  to mean  inclusion in which all the intervening events are moves of Opponent, \ie~$\pol (y\setdif x)\subseteq \setof{-}$. Similarly, $x\subseteq^0y$ signifies an  inclusion in which all the intervening moves are neutral.  However, we shall 
write $x\subseteq^+y$ for inclusion in which the intervening events are either neutral or moves of Player. (The latter choice reflects the fact that neutral moves in a strategy behave as internal moves of Player.)
We say a configuration $x\in\iconf A$ is {\em +-maximal} iff $x$ is maximal in $\iconf A$ w.r.t.~$\subseteq^+$, \ie~the only way that $x$ extends to a larger configuration is through the occurrence of Opponent moves.  
} \end{notation}

\subsection{Operations on games}

We introduce two fundamental  operations on games.

 \subsubsection{Dual}  The {\em dual}, $A^\perp$, of a game $A$,   comprises the same underlying   
 event structure as $A$ but with  a reversal of polarities.  
As mentioned in the introduction, we shall implicitly adopt the view of Player and understand a strategy in a game $A$ as strategy for Player.  A counterstrategy in a game $A$ is a strategy for Opponent in the game $A$, \ie~a strategy (for Player) in the game $A^\perp$.  
 
 \subsubsection{Simple parallel composition} This operation simply juxtaposes two games, and more generally two \esswp.     Let $(A,\leq_A,\Con_A,\pol_A)$ and  $(B,\leq_B,\Con_B,\pol_B)$ be \esswp.  The events of  $A\vvbar B$
    are $(\setof 1\times A) \cup (\setof 2 \times B)$, their polarities unchanged, with the only relations of causal dependency given by 
$(1,a)\leq (1,a')$ iff $a\leq_A a'$ and 
$(2,b)\leq (2,b')$ iff $b\leq_B b'$;  a finite set $X$ of events   is consistent in $A\vvbar B$ iff its components $X_A$  in $A$ and and $X_B$ in $B$ are individually consistent. The unit w.r.t.~simple composition is the  empty \eswp, written $\emptyset$.
We shall adopt the same operation    
for configurations of a game $A\vvbar B$, regarding a configuration $x$ of the parallel composition as $x_A\vvbar x_B$.  

 If we are not a little careful we can run into distracting technical issues through $(A\vvbar B)\vvbar C$ not being strictly the same as $A\vvbar (B\vvbar C)$.  For our purposes it will suffice to adopt the convention that when we write \eg~$A\vvbar B\vvbar C$ the simple parallel composition of three \esswp\, we shall mean the event structure with events
$$
\setof 1\times A\, \cup\,  \setof 2\times  B \,\cup\,  \setof 3\times C\,,
$$
with causal dependency and consistency copied from those of $A$, $B$ and $C$. As in the binary case, we adopt the same notation for configurations and can describe a typical configuration $x$ of $A\vvbar B\vvbar C$ as $x_A\vvbar x_B\vvbar x_C$.  
  
\subsection{Strategies between games}

A strategy {\em from } a game $A$ {\em to} a game $B$ is a strategy in the compound game $A^\perp\vvbar B$.  Of course we shall have to define what it means to be a strategy in a game.  Given another strategy $\tau$ from the  game $B$ to a game $C$, informally we obtain their composition $\tau \scirc \sig$ from $A$ to $C$ by playing the two strategies off against each other in the common game $B$ and hiding the resulting interaction.

The composition of strategies can introduce hidden deadlocks, conflicts and divergences which affect its observable behaviour:

\begin{example}{\rm 
Let  $B$ be the game consisting of two concurrent Player events $b_1$ and $b_2$, and $C$ the game with a single Player event $c$.  We illustrate 
 the composition of two strategies   $\sig_1$  and $\sig_2$ from the empty game $\emptyset$  to $B$,  with  $\tau$ from $B$ to $C$. 
The strategy $\sig_1$  in the game $\emptyset^\perp\vvbar B$ nondeterministically plays $b_1$ or  $b_2$.  The strategy $\sig_2$ also in the game $\emptyset^\perp\vvbar B$ just plays $b_2$. 
The strategy $\tau$ in the game $B^\perp\vvbar C$  does nothing if  just $b_1$ is played and plays the single Player event $c$ of $C$   if $b_2$ is played. The composition $\tau\scirc \sig_1$ in the game $\emptyset^\perp\vvbar C$  may play $c$ or not according as $\sig_1$ plays $b_1$ or $b_2$.   The composition $\tau\scirc \sig_2$ also in the game $\emptyset^\perp\vvbar C$   must play $c$. But the two  compositions $\tau\scirc \sig_1$ and $\tau\scirc \sig_2$ are indistinguishable once the interaction over the common game $B$ is hidden.
 }\endex\end{example}


%

If we are to distinguish the two compositions of the example,  
we need to take some account of their internal moves of interaction.

\section{Strategies with neutral moves}
 
Thus motivated, 
we study {\em bare strategies} with neutral moves, in which we can see the events of interaction not visible in the game.  
Recall we assume that in games all events have +ve or $-$ve polarity.

\begin{definition}\label{def:barestrat}{\rm 
A {\em bare strategy}   from  a game $A$  to a game $B$  comprises a
  total map  $\sig: S\to A^\perp\vvbar N\vvbar B$ of \esswp~(in which $S$ may also have neutral events)
where 
\begin{itemize}
\item[(i)]
 $N$ is an event structure consisting  solely of neutral events;
\item[(ii)] $\sig$ is {\em receptive},\\
$\all x\in\conf S, y\in\conf{A^\perp\vvbar N\vvbar B}.\  \sig x \subseteq^- y \implies \exists! x'\in\conf S.\ x\subseteq x' \ \&\ \sig x' =y$\,;
 \item[(iii)] 
$\sig$ is innocent in  that it is both +-innocent and $-$-innocent:\\
{\em $+$-innocent:}
if
$s\imc s' \ \& \   \pol(s) = +$ then $ \sigma(s)\imc \sigma(s')$\,; \\
{\em $-$-innocent:}
if
$s\imc s' \ \& \   \pol(s') = -$ then $ \sigma(s)\imc \sigma(s')$\,. 
\end{itemize}
Note that $s'$ in +-innocence and $s$ in $-$-innocence may be neutral events.\footnote{This definition of {\em linear innocence}, which applies in the presence of neutral events, appears in the work of Claudia Faggian and Mauro Piccolo~\cite{FP}. It is not to be confused with the innocence of Martin Hyland and Luke Ong, to which it only relates indirectly;  to disambiguate the two notions  ``courtesy'' has been used  for that here.
An extension of Hyland-Ong  innocence to concurrent games is given in~\cite{lics15}.  
 Bare strategies have also been called ``partial'' strategies~\cite{stratsproc} and ``uncovered'' strategies~\cite{essential}.}

 A {\em strategy} from  a game $A$  to a game $B$ comprises a
  total map  $\sig: S\to A^\perp\vvbar B$ of \esswp~for which the composite 
$\sig: S\to A^\perp\vvbar B\iso A^\perp\vvbar \emptyset \vvbar B$ is a bare strategy~\cite{lics11}. 

We shall often identify strategies with bare strategies with no neutral events, and (bare) strategies in a game with (bare) strategies from the empty game $\emptyset$.  
}\end{definition}

Consider two bare strategies $\sig:S\to A^\perp\vvbar N\vvbar B$ and $\sig':S'\to A^\perp\vvbar N\vvbar B$.  A map between them, a 2-cell  $f:\sig\Rightarrow \sig'$,  comprises a map $f:S\to S'$ of \esswp~such that
$$
\xymatrix{
S\ar[d]_\sig\ar[r]^f &S'\ar[dl]^{\sig'}\\
A^\perp\vvbar N\vvbar B&}
$$
commutes.  In this way bare strategies in $A^\perp\vvbar N\vvbar B$  form a category, $$\NStrat(A,N,B)\,.$$   
We obtain the category $\Strat(A,B)$ 
of strategies from $A$ to $B$ from the  special case when $N=\emptyset$.

\subsection{Strategies from bare strategies}

We obtain a strategy as the visible part of a bare strategy
when we hide neutral events via a projection $p$:

\begin{prop}
Let $\sig: S\to A^\perp\vvbar N\vvbar B$ be a bare strategy---so satisfying properties (i), (ii) and (iii) of Definition~\ref{def:barestrat}.  Then, $\sig$ satisfies an additional property:
\begin{itemize}
\item[(iv)]  in the partial-total factorisation of the composition of  $\sig$ with the projection $A^\perp\vvbar N\vvbar B\to A^\perp\vvbar B$,
$$\xymatrix{\ar[d]_\sig \ar[r]^p S&S_\downarrow\ar[d]^{\sig_\downarrow}\\
A^\perp \vvbar N\vvbar B\ar[r] &A^\perp\vvbar B}
$$
  the  defined part $\sig_\downarrow$ is a strategy, which we call the {\em visible part} of $\sigma$. 
\end{itemize}
(Conversely, (iv) together with 
receptivity and no incidence of a +ve event immediately preceding a neutral event in $S$,
 suffice to establish that $\sig$ is a bare strategy.)
\end{prop}
 
With the notation of the lemma above, write
$$\vis x\eqdef px\in\iconf{\vis S}\,$$ for the visible image of a configuration $x\in\iconf S$. 
The hiding operation on strategies extends to a  functor 
$$
\vis{(\_)}: \NStrat(A,N,B) \to \Strat(A,B)\,; 
$$
a 2-cell $f:\sig\Rightarrow \sig'$ between bare strategies restricts to a 
2-cell $\vis{f}: \vis{\sig}\Rightarrow \vis{\sig'}$ between their visible parts. 
It acts so
$$
\vis f \vis x = \vis{(f x)}
$$
for all $x\in\iconf S$.  



\subsection{Composition}
We can compose two bare strategies $$\sig:S\to A^\perp\vvbar M\vvbar  B \ \hbox{  and }\  \tau:T\to  B^\perp\vvbar N\vvbar C$$ by pullback.    Ignoring polarities temporarily,  and padding with identity maps, 
we obtain $\tau\sncirc \sig$ via the pullback  
$$
 \xymatrix{
 &\ar[dl]_{
 }T\sncirc S
 \pb{270}\ar[dr]^{
 }&\\
S\vvbar N\vvbar C\ar[dr]_{\sig\vvbar N \vvbar C}&&\ar[dl]^{A\vvbar M\vvbar \tau} A\vvbar M\vvbar T\\
 &A\vvbar M\vvbar B \vvbar N \vvbar C\,
}$$
as the ensuing map 
$$
\tau\sncirc \sig: T\sncirc S\to A^\perp\vvbar (M \vvbar B^0 \vvbar N)\vvbar  C
$$
once we reinstate polarities and make the events of $B$ neutral.  

As a pullback the configurations of $T\sncirc S$  are built from configurations of $S$ and $T$. 
Let $x\in \iconf S$ and $y\in \iconf T$.  Let $\sig x = x_{A^\perp} \vvbar x_0\vvbar x_B$ and $\tau y = y_{B^\perp} \vvbar y_0\vvbar y_C$.  
Define
$$
y\sncirc x \eqdef (x\vvbar y_0\vvbar y_C) \wedge (x_{A^\perp} \vvbar x_0\vvbar y)
$$
which will be defined and a configuration in $\iconf{T\sncirc S}$  if $x_B=y_{B^\perp}$ and the corresponding bijection secured. 
The following property, useful later,  is a consequence of the receptivity of $\sig$ and $\tau$.
\begin{lemma}\label{lem:maxlforpstrats}  
A configuration $y\sncirc x\in \iconf{T\sncirc S}$ is 
+-maximal in $\iconf{T\sncirc S}$ iff $x$ is +-maximal in $\iconf S$ and $y$ is +-maximal in $\iconf T$.
\end{lemma}

Given a 2-cell $f:\sig\Rightarrow \sig'$ between bare strategies in $A^\perp\vvbar N\vvbar B$
and $g:\tau\Rightarrow \tau'$ between bare strategies in $B^\perp\vvbar M\vvbar C$, from the universality of pullback we obtain the 2-cell
$$g\sncirc f: \tau\sncirc \sig \Rightarrow \tau'\sncirc \sig'\,$$
between the two compositions in $A^\perp\vvbar (N\vvbar B^0\vvbar  M)\vvbar C$. 
It acts so 
$$
(g\sncirc f)(y\sncirc x) = (g y)\sncirc (fx)
$$
on a typical configuration  $y\sncirc x$. 
This extends  composition of bare strategies to a functor
$$\sncirc: \NStrat(B, N,  C) \times \NStrat(A,M,B) \to  \NStrat(A,M \vvbar B^0 \vvbar N,C)\,.
$$ 
Composition  of bare strategies restricts to a functor between rigid 2-cells.

We obtain the composition of strategies as the composite functor
$$
\eqalign{
\scirc: \Strat(B,C)\times \Strat(A,B) \iso\  &
\NStrat(B,\emptyset, C)\times \NStrat(A,\emptyset, B)\cr
&\arr{\sncirc}
 \NStrat(A,\emptyset\vvbar B^0\vvbar \emptyset, C)\arr{\vis{(\_)}} \Strat(A,C)\,. }
$$
Though we generally elide the isomorphisms regarding strategies as bare strategies without neutral events, and write 
$$
\tau \scirc \sig \eqdef \vis{(\tau \sncirc \sig)}
$$ 
for the composition of strategies $\sig\in \Strat(A, B)$ and $\tau\in\Strat(B,C)$. 
Describing the strategies $\sig:S\to A^\perp\vvbar B$ and $\tau: T\to B^\perp\vvbar C$ as having composition 
$$\tau\scirc \sig: T\scirc S \to A^\perp\vvbar C\,,$$ we can
present a typical configuration of $T\scirc S$ as 
$$y\scirc x \eqdef (y\sncirc x)_\downarrow$$ 
for $x\in\iconf S$ and $y\in\iconf T$.

Composition is preserved in extracting the visible part from bare strategies: 
\begin{lemma}\label{lem:visprescomp}\label{lem:interactdefpart}
Let  $\sig:S\to A^\perp\vvbar  M\vvbar  B$ and $\tau: T\to B^\perp \vvbar N\vvbar C$ be bare strategies.  Then,
$$
\vis{(\tau \sncirc \sig)} = \vis{\tau} \scirc \vis{\sig}\,.
$$
For $x\in\iconf S$ and $y\in\iconf T$, 
$$
(y\sncirc x)_\downarrow  = y_\downarrow \scirc x_\downarrow\,,
$$
with  
one side defined if the other is.
\end{lemma}  
 
\subsection{The copycat strategy}
 The copycat strategy is the identity for composition of strategies. We present its construction and key property from~\cite{lics11}.

\begin{lemma}\label{lem:copycat}
Let $A$ be an \eswp.  
There 
is an \eswp\, $\CC_A$ having the same events  and polarity as ${A^\perp\vvbar A}$ but with causal dependency $\leq_{\CC_A}$ given as the transitive closure of the relation
$$\leq_{A^\perp\vvbar A} 
\cup\ 
\set{(\bar c,  c)}{c\in {A^\perp\vvbar A} \ \&\ \pol_{A^\perp\vvbar A}(c) = +}\,
$$  
and  finite subsets of $\CC_A$ consistent if their down-closure w.r.t.~$\leq_{\CC_A}$ are consistent in ${A^\perp\vvbar A}$.  (For $c\in A^\perp\vvbar A$ we use $\bar c $ to mean the corresponding copy of $c$, of opposite polarity,  in the alternative component, \ie~$\overline{(1,a)} = (2, a) \hbox{ and } \overline{(2,a)} = (1, a)$.)

\noindent
The configurations of $\CC_A$ have the form $x\vvbar y$ where $y\sqsubseteq_A x$, \ie~$y\supseteq^- x\cap y \subseteq^+ x$,  for $x, y\in\iconf A$. (The relation $\sqsubseteq_A$  is a partial order, called the {\em Scott order}~\cite{fossacs13}.)

\noindent
 The {\em copycat} strategy for $A$ is the map $\cc_A: \CC_A\to A^\perp\vvbar A$ which acts as identity on events. 
We have  $\sig \iso \sig\scirc \cc_A \iso \cc_B\scirc \sig$,
for any strategy $\sig\in\Strat(A,B)$. 
\end{lemma} 

The axioms on strategies are precisely those needed to ensure that copycat behaves as identity w.r.t.~composition $\scirc$ and thus obtain a bicategory $\Strat$ of games and strategies~\cite{lics11}.
Of course copycat is not the identity for the composition of bare strategies; that composition will generally have extra neutral events introduced through interactions.

\section{`May' and `must' tests}

Consider the following three bare strategies in the game $A$ comprising a single Player move $\plmove$.   Neutral events are drawn as $\nemove$. 
 
$
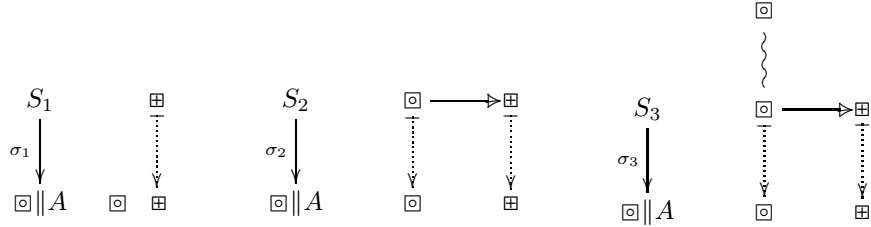
\xymatrix{
& 
\\
 S_1\ar[d]_{\sig_1}&  
 {{\plmove}}\ar@{|.{>}}[d] \\
\nemove\vvbar A  &\!\!\!\!\!\!\!\!\!\!\nemove\quad  \plmove  
 }
 $ 
  $\xymatrix{
&& 
 & 
\\
& S_2\ar[d]_{\sig_2}& \nemove \ar@{|.{>}}[d]  \ar@{|>}[r] 
 &\plmove\ar@{|.{>}}[d] \\
&\nemove\vvbar A& \nemove & \plmove  
 }
 $ 
$\xymatrix{
&& 
\nemove  \ar@{~}[d]
& 
\\
& S_3\ar[d]_{\sig_3}& \nemove  \ar@{|>}[r] \ar@{|.{>}}[d]
 &\plmove\ar@{|.{>}}[d] \\
&\nemove\vvbar A& \nemove& \plmove  
 }
 $ 

\noindent
From the point of view of observing the move over the game $A$ the first two bare strategies, $\sig_1$ and $\sig_2$, differ from the the third, $\sig_3$.  In a maximal play both $\sig_1$ and $\sig_2$ must result in the observation of the single move of $A$.  However, in $\sig_3$ one maximal play is that in which the topmost neutral event of $S_3$ has occurred, in conflict with the only way of observing the single move of $A$.  

We follow  Hennessy and  de Nicola in making these ideas precise~\cite{hennessy-denicola}.  

 \begin{definition}{\rm  Let $\sig$ be a bare strategy in a game $A$.  Let $\tau:T\to A^\perp\vvbar N\vvbar  \plmove$ be a `test' bare strategy from $A$ to the game consisting of a single Player move $\plmove$. Write $\tick\eqdef (3, \plmove)$.    

Say $\sig$ {\em may pass}  $\tau$ iff there exists  $y\sncirc x\in \iconf{T\sncirc S}$, where $x\in\iconf S$ and $y\in\iconf T$,  with the image $\tau  y$ containing $\tick$. (Note that we may w.l.o.g.~assume that the configuration $y\sncirc x$ is finite.) 

Say $\sig$ {\em must pass}  $\tau$ iff  for all $y\sncirc x\in \iconf{T\sncirc S}$, where $x\in\iconf S$ and $y\in\iconf T$ are $\subseteq^+$-maximal, the image  $\tau y$ contains $\tick$.

Say two bare strategies are {\em `may'} (respectively, { \em`must')}  {\em equivalent} iff the  tests they may (respectively, must) pass are the same.  
 
The definitions extend in the obvious fashion to bare strategies of type $A^\perp\vvbar N \vvbar B$.
}\end{definition}

A bare strategy is `may' equivalent, but need not be `must' equivalent, to the  strategy which is its defined part; 
 `must' inequivalence is lost in moving from bare strategies to strategies.  

\begin{example}\label{ex:musttesting}{\rm
As an illustration of the subtle nature of testing for `must' equivalence, consider the following bare strategies in the game $A$, as drawn:
$$
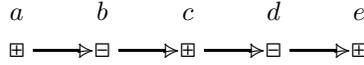
\xymatrix@R=2pt@C=20pt{a&b&c&d&e\\
\plmove\ar@{|>}[r] &\opmove\ar@{|>}[r] &\plmove\ar@{|>}[r] &\opmove\ar@{|>}[r] &\plmove
}
$$
The game $A$ consists of five events with the polarity and causal dependency shown.  
One bare strategy $\sig_1$ is
$$
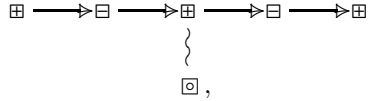
\xymatrix@R=16pt@C=20pt{
\plmove\ar@{|>}[r] &\opmove\ar@{|>}[r] &  \ar@{~}[d]\plmove\ar@{|>}[r] &\opmove\ar@{|>}[r] &\plmove\\
&&\nemove\,,\!\!\!\!}
$$
with one neutral event,
while the other $\sig_2$ is
$$
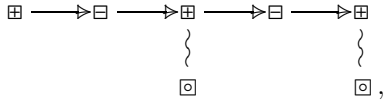
\xymatrix@R=16pt@C=20pt{
\plmove\ar@{|>}[r] &\opmove\ar@{|>}[r] &  \ar@{~}[d]\plmove\ar@{|>}[r] &\opmove\ar@{|>}[r] & \ar@{~}[d]\plmove\\
&&\nemove&&\nemove\,,\!\!\!\!}
$$
with two neutral events.
Through the possible occurrence of neutral events the bare strategy $\sig_1$ has a +-maximal configuration with just events $\setof{a,b}$ visible from the game, while $\sig_2$ has in addition a +-maximal configuration comprising visible moves  $\setof{a,b,c,d}$.  The following test strategy distinguishes $\sig_1$ and $\sig_2$
$$
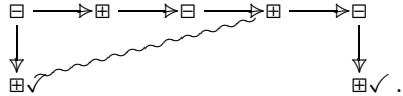
\xymatrix@R=16pt@C=20pt{
\opmove\ar@{|>}[r] \ar@{|>}[d]&\plmove\ar@{|>}[r] & \opmove\ar@{|>}[r] &\plmove\ar@{|>}[r] &\opmove\ar@{|>}[d]\\
\plmove  \ar@{~}[urrr]\tick\!\!\!\!\!\!&&&&\plmove\tick\,.\!\!\!\!\!\!\!\!\!}
$$
in that $\sig_1$ must pass the test while $\sig_2$ need not. (The test is a strategy in $A^\perp\vvbar  \plmove\tick$; thus the change in polarity of moves in $A$.)  Composed with the test,   $\sig_1$ may fail to perform the leftmost $\plmove\tick$ but if so must then perform the rightmost $\plmove\tick$; whereas $\sig_2$   may fail to perform both.  
}\endex\end{example}

This example might be puzzling to readers familiar with Harmer and
McCusker's fully abstract model for (sequential) finite non-determinism~\cite{HarmerM99}, in particular regarding their handling of `must' equivalence through
the addition of {\em divergences}. Indeed, there, if a trace
has already potentially triggered a divergence, then any execution going
past that divergence has already lost all hope for `must' convergence.
Accordingly, in~\cite{HarmerM99}, only the first divergence matters---further
divergences are not recorded.

In contrast, Example~\ref{ex:musttesting} indicates how in concurrent strategies, all divergences
matter. This is not an artificiality of concurrent strategies, but something
inherent to the observational power of the tests they support: here the
distinguishing tests manage to observe beyond the first divergence by
running two threads in parallel. The first thread aims to directly
register success; by +-maximality it will be run eventually if the observed
strategy triggers the first divergence. The second thread follows the
execution of the program, then immediately cancels the first thread if
the program goes past the divergence; and then proceeds to test for the
second divergence.

\section{Strategies with stopping configurations}

Bare strategies lack identities w.r.t.~composition, so they do not 
form a bicategory.  Fortunately, 
for `may' and `must' equivalence it is not necessary to use bare strategies; for `may' equivalence strategies suffice; whereas for `must' equivalence  it is sufficient  to carry with a strategy the extra structure of {\em stopping 
configurations}---to be thought of as images of $+$-maximal configurations in an underlying bare strategy.  
As we shall see,
composition and copycat extend to composition and copycat on strategies with stopping configurations, while maintaining  a bicategory.  
We  
tackle the simpler case in which games are assumed to be race-free.  (The extension to games which are not race-free is outlined in~\cite{stratsproc}.)  We recall when an \eswp~is race-free  and the allied notion of deterministic strategy:

\begin{definition}{\rm 
Say $A$, an \eswp, is {\em race-free} iff whenever 
$x\subseteq^+ y$ and $x\subseteq^- z$
for configurations $x$, $y$, $z$ of $A$ then 
$y\cup z$ 
 is also a configuration. 

Say $S$, an \eswp, is {\em deterministic} iff whenever  
$x\subseteq^+ y$ and $x\subseteq z$
for configurations $x$, $y$, $z$ of $S$ then  
$y\cup z$ 
 is also a configuration.  Say a bare strategy $\sig: S\to A^\perp\vvbar N\vvbar B$ is {\em deterministic} iff $S$ is deterministic; with a strategy being deterministic iff it is so as a bare strategy.
}\end{definition}

\begin{lemma}\cite{lics11,DBLP:journals/fac/Winskel12}
Let $A$ be a game.  The copycat strategy $\cc_A$ is deterministic iff the game $A$ is race-free.
\end{lemma}

Let $\sig:S\to A^\perp\vvbar  N\vvbar B$ be a bare strategy between race-free games  $A$ and $B$. Recall its associated partial-total factorisation
$$
\xymatrix{
\ar[d]_\sig \ar[r]^{p} S& \vis S\ar[d]^{\vis{\sig}}\\
 A^\perp\vvbar N\vvbar B\ar[r] &A^\perp\vvbar B}
$$
where $p$ is a projection sending configurations $x$  of $S$ to configurations  $\vis x$ of $\vis S$.
The visible part  of $\sig$ is a strategy $\vis{\sig}$. 
Define the {\em  stopping}  configurations in $\iconf{S_\downarrow}$ to be 
$$
\Stop(\sig)\eqdef  
\set{\vis x}{x \in\iconf S  \hbox{ is $+$-maximal}}\,.
$$
 So, 
in other words, the stopping configurations are the visible  images of  configurations which are maximal w.r.t.~neutral or Player moves. 
Note that $\Stop(\sig)$ will include all the +-maximal configurations of $S_\downarrow$: any +-maximal configuration $y$ of  $S_\downarrow$  is the image under $p$ of its down-closure $[y]$ in $S$, and by Zorn's lemma this extends (necessarily by neutral events) to a maximal configuration $x$ of $S$ with image $y$ under $p$; 
the configuration $x$ is +-maximal by the +-maximality of $y$. 
If $\sig$ is deterministic, then $\Stop(\sig)$ consists of precisely the +-maximal configurations of $\vis S$.  
If $\sig$ is a strategy, \ie~it has no neutral events, then $\Stop(\sig)$ is just the set consisting of all +-maximal configurations of $S$.

\begin{definition}\label{def:stratwstpping}{\rm 
A {\em stopping strategy} in a game $A$ comprises  $(\sig, M_S)$, a strategy $\sig:S\to A$ together with a subset $M_S\subseteq\iconf S$ called {\em stopping configurations}.  
As usual, a stopping strategy  from a game $A$ to  game $B$ is a stopping strategy in the game $A^\perp\vvbar B$. 
We let
 $\St: \sig\mapsto (\vis{\sig}, \Stop(\sig))$  denote the operation motivated above from bare strategies $\sig$ to stopping strategies.  
}\end{definition}

\noindent
{\bf Remark.}
There is the issue of what axioms to adopt on stopping configurations.  We do not insist that stopping configurations include all +-maximal configurations as this property will not be preserved in taking the rigid image of a stopping strategy---Example~\ref{ex:rigim}. 
\\

Given two stopping strategies   $\sig:S\to A^\perp\vvbar B$, $M_S$ and $\tau: T\to B^\perp\vvbar C$, $M_T$ 
we define their {\em interaction},  
$$
(\tau, M_T)\sncirc (\sig, M_S) \eqdef (\tau\sncirc \sig, M_T\sncirc M_S)\,,
$$
with the stopping configurations of the interaction $T\sncirc S$ as
$$
M_T\sncirc M_S = \set{y\sncirc x}{x\in M_S \ \&\ y\in M_T}
$$
---sensible because of Lemma~\ref{lem:maxlforpstrats}.
Define their 
{\em composition}    by 
$$
(\tau, M_T)\scirc (\sig, M_S) \eqdef (\tau\scirc \sig, M_T\scirc M_S)\,,
$$
where 
 the stopping configurations of $T\scirc S$ form the set
$$
M_T\scirc M_S = \set{y\scirc x}{x\in M_S \ \&\ y\in M_T}\,.
$$

To make stopping strategies into  a bicategory we must settle on an appropriate notion of 2-cell. The following choice of definition is useful for `must' equivalence---see Lemma~\ref{lem:2-cellsmaymust}.

\begin{definition}{\rm  A 2-cell  $f: (\sig, M_S) \Rightarrow (\sig', M_{S'})$ between stopping strategies is a 2-cell of strategies $f:\sig\Rightarrow \sig'$ such that $fM_S \subseteq M_{S'}$.  We write $\SStrat(A,B)$ for the category of stopping strategies from game $A$ to game $B$; its maps are 2-cells.  
}\end{definition}

  Composition extends to 2-cells between stopping strategies: 
 their composition as 2-cells between strategies is easily shown to preserve stopping configurations. 

\begin{prop} For games   $A$, $B$ and $C$  composition of stopping strategies is a functor
$\scirc: \SStrat(B,C)\times \SStrat(A,B) \to \SStrat(A,C)\,.$
  \end{prop}

We should also extend copycat $\cc_A:\CC_A\to A^\perp\vvbar A$  to a stopping strategy.  Because we are assuming $A$ is race-free,  we do this by taking 
$$
M_{\CC_A} \eqdef \set{(x\vvbar x)  \in \iconf{\CC_A}}{ x\in\iconf A}\,.
$$
Because $A$ is race-free, $M_{\CC_A}$ comprises all the +-maximal configurations of $\CC_A$. 
Then, $(\cc_A, M_{\CC_A})$ is an identity w.r.t.~the extended composition.  


With the operations and constructions above, stopping strategies inherit the structure of a bicategory $\SStrat$ from strategies; the objects are restricted to race-free games in order to have the above simple form of stopping configurations for copycat. 
  
\subsection{Bare strategies and stopping strategies}
We turn to relations between bare strategies and stopping strategies.  Recall from Definition~\ref{def:stratwstpping}
the operation $$\St: \sig\mapsto (\vis{\sig}, \Stop(\sig))$$  which takes a bare strategy $\sig$ to a stopping strategy.  It preserves composition:

\begin{lemma}\label{lem:Stprescomp}
Let $\sig:S\to A^\perp\vvbar  M \vvbar B$  and $\tau:T\to B^\perp\vvbar N\vvbar C$ be bare strategies.  Then, 
$$
\St(\tau \sncirc \sig) = \St(\tau) \scirc \St(\sig)\,.
$$
\end{lemma}
\begin{proof}  
%
By Lemma~\ref{lem:visprescomp}, 
it suffices to show 
$$
\Stop(\tau \sncirc \sig) = \Stop(\tau) \scirc \Stop(\sig)\,.
$$
Configurations of $\Stop(\tau \sncirc \sig)$ are of the form $(y\sncirc x)_\downarrow$ where $y\sncirc x$ is +-maximal for
$x\in\iconf S$ and $y\in \iconf T$.  By Lemma~\ref{lem:maxlforpstrats}
these coincide with configurations $(y\sncirc x)_\downarrow$ for which both $x$ and $y$ are +-maximal.  
Configurations of $\Stop(\tau) \scirc \Stop(\sig)$ take the form $y_\downarrow \scirc x_\downarrow$ where $x\in\iconf S$  and   $y\in \iconf T$ are  +-maximal.  But by Lemma~\ref{lem:interactdefpart}, 
$(y\sncirc x)_\downarrow= y_\downarrow \scirc x_\downarrow$, 
ensuring the desired equality.
\end{proof}

We have seen that there is a functor $\vis{(\_)}: \NStrat(A,N,B)\to \Strat(A,B)$  from bare strategies to their visible part. However, 
it is not the case that $\St$ is a functor from bare strategies $\NStrat(A,N,B)$ to stopping strategies $\SStrat(A,B)$.  Given an arbitrary  2-cell  $f: \sig\Rightarrow \sig'$ between bare strategies $\vis{f}$ can fail to preserve 
stopping configurations.  However:

\begin{prop} Let $\sig:S\to A^\perp\vvbar  N\vvbar B$ and $\sig':S'\to A^\perp\vvbar N\vvbar B$ be bare strategies.  Say a 2-cell  $f: \sig\Rightarrow \sig'$ is {\em +-reflecting} iff, 
for $x\in\conf S$, $y\in\conf{S'}$, 
$$
fx\subseteq^+ y \implies \exists x'\in\conf S.\ x\subseteq x' \ \&\ fx' = y\,.
$$
Let  $\NStrat^+(A,N,B)$ be the subcategory where 2-cells are +-reflecting. Then
$$
\St: \NStrat^+(A,N,B) \to \SStrat(A,B)
$$
taking  $f: \sig\Rightarrow \sig'$ to $\vis{f}: \St(\sig)\Rightarrow \St(\sig')$ is a functor.  
\end{prop}

\section{`May' and `Must' testing}

We can  rephrase `may' and `must' testing in terms of stopping strategies.  

 \begin{definition}{\rm  Let $(\sig, M_S)$ be a stopping strategy in a game $A$.  Let $\tau:T\to  A^\perp\vvbar N\vvbar \plmove$ be a `test' bare strategy from $A$ to a the game consisting of a single Player move $\plmove$. Write $St(\tau)$ as $(\tau_0, M_0)$ where 
$\tau_0:T_0\to A\vvbar \plmove$ is the visible part of $\tau$ and $M_0$ are its stopping configurations, obtained as images of the $+$-maximal configurations of $T$. Write $\tick\eqdef (2, \plmove)$.   

Say $(\sig, M_S)$ {\em may pass}  $\tau$ iff there exists  $y\sncirc x\in \iconf{T_0\sncirc S}$, where $x\in\iconf S$ and $y\in\iconf{T_0}$,  with the image $\tau_0  y$ containing $\tick$. (Note again, we may w.l.o.g.~assume that the configurations  $x$  and $y$ are finite.) 

Say $(\sig, M_S)$ {\em must pass}  $\tau$ iff  for all $y\sncirc x\in M_0\sncirc M_S$, where $x\in\iconf S$ and $y\in\iconf{T_0}$,  the image  $\tau_0 y$ contains $\tick$. 

Say two stopping strategies  are {\em `may'}, respectively {\em `must'}, {\em equivalent} iff the  tests they may, respectively must, pass are the same.   
}\end{definition}

\begin{prop} With the notation above,

$(\sig, M_S)$ {\em may pass}  $\tau$ iff  there exists  $y\scirc x\in \iconf{T_0\scirc S}$, where $x\in\iconf S$ and $y\in\iconf{T_0}$,  with the image $\tau_0 y$ containing $\tick$ ---the configurations $x$, $y$ may be assumed finite; and 

$(\sig, M_S)$ {\em must pass}  $\tau$ iff  for all  $y\scirc x\in M_0\scirc M_S$, where $x\in M_S$ and $y\in M_0$,  the image  $\tau_0\ y$ contains $\tick$.  
\end{prop}

\begin{lemma} Let $A$ be a race-free game. Let $\sig$ be a bare strategy in $A$.  
Then,

$\sig$ {\em may pass a test}  $\tau$ iff $\St(\sig)$   {\em may pass}  $\tau$; 

$\sig$ {\em must pass a test}  $\tau$ iff $\St(\sig)$   {\em must pass}  $\tau$.
\end{lemma}
\begin{proof}
Directly from the definitions, for the  `if' of the `must' case, using Lemma~\ref{lem:maxlforpstrats}. 
\end{proof}

\begin{example}{\rm It is tempting to think of neutral events as behaving like the internal ``tau''  events of CCS~\cite{milner}.  However,  in the context of concurrent strategies, 
because of their asynchronous nature,  they behave rather differently. 
Consider 
   three bare strategies, over a game comprising of just two concurrent $+$ve events, say $a$ and $b$.  The bare strategies have  the following event structures  in which we have named events by the moves they correspond to in the game:
 
$\xymatrix@R=12pt@C=8pt{
 S_1&  a\ar@{~}[d]\\
&b}
$ 
\qquad\quad
$\xymatrix@R=12pt@C=12pt{
 S_2&  \nemove\ar@{~}[d]\ar@{|>}[r] &a\\
&\nemove\ar@{|>}[r] &b}
$ 
\qquad\quad
$\xymatrix@R=12pt@C=12pt{
 S_3&  \nemove\ar@{~}[d]\ar@{|>}[r] &a\\
&b &}
$ 
  
\noindent
No pair would be weakly bisimilar due to the presence of pre-emptive internal events~\cite{milner}.  However, 
all three become isomorphic under  $\St$ so are `may' and `must' equivalent to each other. }\endex
\end{example}

2-cells between stopping strategies  respect `may' and `must' behaviour in the sense of the  following lemma.

\begin{lemma}\label{lem:2-cellsmaymust}
Let $f: (\sig, M_S) \Rightarrow (\sig', M_{S'})$ be a 2-cell between stopping strategies. Then for any  test $\tau$,

$(\sig, M_S)$ may pass $\tau$ implies $(\sig', M_{S'})$ may pass $\tau$; and  

$(\sig', M_{S'})$ must pass $\tau$ implies $(\sig, M_S)$ must pass $\tau$. 

Moreover, if $f$ is a rigid epi and $fM_S = M_{S'}$, then $ (\sig, M_S)$ and $(\sig', M_{S'})$ are both `may' and `must' equivalent.
\end{lemma}
 \begin{proof} 
In this proof, we shall identify a test with its image $(\tau, M_T)$ under $\St$, as a strategy  $\tau:T\to A^\perp\vvbar \plmove$   with stopping configurations   $M_T$. 

Let $f: (\sig, M_S) \Rightarrow (\sig', M_{S'})$ be a 2-cell.  Assume $\sig:S\to A$ and $\sig':S'\to A$.

Suppose $(\sig, M_S)$ may pass $(\tau, M_T)$.   Then there is a (finite) configuration which we write
$y \sncirc x$ of $T\sncirc S$, built as a secured bijection out of $y\in\conf T$ and $x\in\conf S$,
whose image in the game contains $\tick$.   The secured bijection built out of $y$ and $x$ induces a secured bijection built  out of $y$ and $fx$; this is because $fx$ has no more causal dependency than $x$ with which it is in bijection. This  determines a configuration $y\sncirc fx$, with image containing $\tick$. 

Suppose $(\sig', M_{S'})$ must pass $(\tau, M_T)$.  Any $y\sncirc x\in M_T\sncirc M_S$ images under $\tau\sncirc f$ to $y\sncirc fx\in M_T\sncirc M_{S'}$.  As  $(\sig', M_{S'})$ must pass $\tau$, the configuration $y\sncirc fx$ has image containing $\tick$, ensuring that $y\sncirc x$ does too.   

Finally suppose that $f$ is rigid epi and $fM_S = M_{S'}$.  We have just shown that $f$ preserves the passing of `may' tests and reflects the passing of `must' tests.  Because $f$ is rigid epi it also reflects the passing of `may' tests.  Because $f$ is rigid and $fM_S = M_{S'}$ it preserves the passing  of `must' tests: any secured bijection $y\sncirc fx$ in $M_T\sncirc M_{S'}$ ensures by the rigidity of $f$ a secured bijection $y\sncirc x$ in $M_T\sncirc M_S$; as $(\sig, M_S)$ must pass $\tau$ we have the image in the game of $y\sncirc x$  contains $\tick$ ensuring the image of  $y\sncirc fx$ does too. 
\end{proof}

Tests based on bare strategies are more discriminating than tests based on (pure) strategies:

\begin{example}\label{ex:ptestsmoredescrim}{\rm   Let a game comprise a single Player move.  Consider two stopping strategies:

$\sig_1$, the empty strategy with the empty configuration $\emptyset$  as its single stopping configuration;

$\sig_2$, the strategy performing the single Player move $\plmove$ with stopping configurations  $\emptyset$ and $\setof\plmove$. (We can easily realise this stopping strategy via $\St$ from a bare strategy with event structure  $\xymatrix@R=2pt@C=12pt{\plmove\ar@{~}[r]&\nemove}$.)

By Lemma~\ref{lem:2-cellsmaymust}, we have
$(\sig_2, 
\setof{\emptyset, \setof\plmove})$ must pass $\tau$ implies $(\sig_1, \setof{\emptyset})$ must pass $\tau$,
for any test $\tau$.  (The above would not hold if we had not included $\emptyset$ in the stopping configurations of $\sig_2$.
) 

Using the fact that we need only consider rigid images of tests---shown later in Section~\ref{sec:axioms}, a little argument by cases establishes the converse implication too, provided we restrict just to tests which are strategies.  
The stopping strategies would be must equivalent w.r.t.~tests based just on strategies.  

However with tests based on bare strategies we can distinguish them. Consider the test $\tau$ comprising three events, one of them neutral, with only nontrivial causal dependency
$\opmove \imc\nemove$ and 
$\nemove$ in conflict with the `tick' event $\plmove$.  Then, it is not the case that $(\sig_2, 
\setof{\emptyset, \setof\plmove})$ must pass $\tau$ ---the occurrence of the neutral event blocks success in a maximal execution---while $(\sig_1, \setof{\emptyset})$ must pass $\tau$. Notice how
the presence of the neutral event in the test turns the possibility that $\sig_2$ can perform the Player move into a possibility of its failing the test.
}\endex\end{example}

\section{`May' and `Must' behaviour characterised}

\subsection{Preliminaries, traces of a strategy}

Let $S$ be an event structure.  
A possibly infinite sequence
$$
s_1, s_2, \cdots, s_n, \cdots 
$$ 
in $S$ constitutes a {\em serialisation}  of a configuration $x\in \iconf S$ if $x=\setof{s_1, s_2, \cdots, s_n, \cdots}$ and $\setof{s_1, \cdots, s_i}\in\conf S$ 
for all $i$ at which the sequence is defined. We will often identify such a countable enumeration of a set with its associated  total order.  Note that in this way we can regard a serialisation as  an elementary event structure in which causal dependency takes the form of a total order; a serialisation of a configuration is associated with a map to $S$ whose image is the configuration.

Let $\sig:S\to A$ be  a strategy in a game $A$.  A {\em trace} in $\sig$ is  a possibly infinite sequence
$$
\al =(\sig(s_1), \sig(s_2), \cdots, \sig(s_n), \cdots )
$$
of events in $A$ 
obtained from a serialisation  
$$
s_1, s_2, \cdots, s_n, \cdots 
$$ 
of a configuration $x\in\iconf S$.  Clearly $\al$ is a serialisation of  $\sig x\in\iconf A$.  
From the local injectivity of $\sig$,  the configuration $x$ will be finite/infinite according as the trace is finite/infinite.  
We say that $\al$ is a trace of the configuration $x$ in $\sig$, or that $x$ has trace $\al$ in $\sig$.  
 
\begin{prop}\label{prop:tracelem}
 Let $\sig:S\to A$ be a strategy.

\noindent
(i) Any countable configuration of $S$ has a trace.

\noindent
(ii) Let $x\in\iconf S$ and $\al$ be an  enumeration 
$$
a_1, a_2, \cdots, a_n, \cdots 
$$
of $\sig x$.  Then, 
$\al$ is a trace of $x$ 
 in $\sig$
iff 
for all $s,s'\in x$ if $s\imc s'$ then $\sig(s)$ precedes 
$\sig(s')$ in the enumeration $\al$.
\end{prop}
\begin{proof}
(i) Let $x$ be a countable configuration of $S$ w.r.t.~the strategy $\sig:S\to A$.
This follows because there is a serialisation $x=\setof{s_1, s_2, \cdots, s_n, \cdots }$, 
in which  $\setof{s_1, \cdots, s_i}$ is down-closed in $S$ at all $i$ in the enumeration.  To see this, from its countability we may 
assume a countable enumeration of $x$, which need not be a serialisation.
Define $s_1\in x$ to be the earliest  event of the enumeration for which $[s_1)=\emptyset$ in $S$; such an $s_1$ is ensured to exist by the well-foundedness of causal dependency provided $x\neq\emptyset$.  Inductively, define $s_n$ to be the earliest  event of the enumeration which is in $x\setdif \setof{s_1, \cdots, s_{n-1}}$ and for which $[s_n)\subseteq \setof{s_1, \cdots, s_{n-1}}$; again the well-foundedness of causal dependency ensures such an $s_n$ exists provided $x\setdif \setof{s_1, \cdots, s_{n-1}}\neq \emptyset$. It is elementary to check this provides a serialisation of $x$.  \\

\noindent
(ii) {\it ``Only if'':} 
 Directly from the definition of trace of a configuration. { \it ``If'':}   Via the local bijection between $x$ and $\sig x$  given by $\sig$ we obtain an enumeration 
$$
s_1, s_2, \cdots, s_n, \cdots 
$$ 
of $x$ matching $\al$ in that $\sig(s_i) = a_i$.   The assumption that $s\imc s'$ implies $\sig(s)$ precedes 
$\sig(s')$ in the enumeration $\al$, entails $\setof{s_1, \cdots, s_i}\in\conf S$ 
for all $i$.  Hence the enumeration of $x$ is a serialisation making $\al$ a trace of $x$.   
\end{proof}
 
\begin{lemma}\label{lem:disagreement-lem}
Let   $\sig:S\to A$ be a strategy in a game $A$. Let $x\in\iconf{S}$.
Let $\al$ be a serialisation  of $\sig x$
which is not a trace of $x\in\iconf{S}$.  Then, there are 
$s, s'\in x$ with
$\pol(s) = -$ and $\pol(s') = +$
and 
$s\imc_S s'$  with (note the order reversal)
$\sig(s')  \leq_{\al}\sig(s)$   in $\al$ (regarded as a total order).
\end{lemma}
\begin{proof}
By assumption, any trace of $x$ differs from $\al$. 
We deduce there is 
$s\imc s'$ in $x$    with $\sig(s) \not\leq \sig(s')$ in the total order of $\al$; otherwise we could serialise $x$ to obtain the   trace $\al$ ---Proposition~\ref{prop:tracelem}(ii). 
Now, $\sig(s) \not\leq_A \sig(s')$ in  $A$ as any serialisation must respect the order $\leq_A$.  
Hence, by the innocence of $\sig$, we must have $\pol(s) = -$ and $\pol(s') = +$.  
Because $\al$ is totally ordered,
$\sig(s')  \leq\sig(s)$   in $\al$.
\end{proof}

\subsection{Characterisation of  the `may' preorder}

For stopping strategies  (with games assumed race-free) we have:

\begin{theorem} Let $(\sigma_1, M_1)$ and  $(\sigma_2, M_2)$ be stopping strategies  in a common game.  Then,

$(\sigma_1, M_1)$ may pass $\tau$ implies  $(\sigma_2, M_2)$ may pass $\tau$, for all tests $\tau$,

iff

all finite traces of $\sigma_1$ are traces of  $\sigma_2$.
\end{theorem}
\begin{proof}  Assume strategies $\sig_1:S_1\to A$ and  $\sig_2:S_2\to A$.
{\it ``if'':}
Assume all finite traces of   $\sig_1$ are traces of  $\sig_2$. Suppose $(\sigma_1, M_1)$ may pass test $\tau$  with event structure $T$.   Then there is a successful  configuration $w\sncirc x_1\in \conf{T\sncirc S_1}$, where $x_1\in \conf{S_1}$ and $w\in\conf T$; it is successful in the sense that its image contains the success event $\tick$.   Take a serialisation of $w\sncirc x_1$; this induces a serialisation of
$x_1$ to yield a trace.  Then, by assumption,  $\sig_2$  has a configuration $x_2\in\conf{S_2}$ with the same trace, so a matching serialisation.  Consequently the pairing $w\sncirc x_2$ is defined with $w\sncirc x_2\in   \conf{T\sncirc S_2}$; sharing the same image as $w\sncirc x_1$ it is also successful.  \\

\noindent
{\it ``only if'':}  We show the contraposition: assuming not all  traces of   $\sig_1$ are traces of $\sig_2$, we produce a test $\tau$  for which 
 $\sigma_1 $ may pass $\tau$ while  it is not the case that $\sigma_2$ may pass $\tau$.  

Assume a trace $\al_1$ of $x_1\in \conf{S_1}$ is not a trace of any $x_2\in \conf{S_2}$.  Note that the trace $\al_1$, and correspondingly $x_1$, must have at least one +ve event as otherwise, by receptivity, $\sig_2$ could match the trace $\al_1$.   
Any trace of $x_2
$, with $\sig_2 x_2 = \sig_1 x_1$,  differs from $\al_1$. 
By Lemma~\ref{lem:disagreement-lem},   we deduce there are  $s, s'\in x_2$ such that $s\imc_2 s'$  with $\pol(s) = -$ and $\pol(s') = +$ and 
$\sig_2(s')  \leq_1\sig_2(s)$   in the total order $\al_1$.


Thus for each  $x_2\in \conf{S_2}$ with $\sig_2 x_2 = \sig_1 x_1$ we can choose $\theta(x_2) = (s,s')$ 
so that $s\imc_2 s'$ in $x_2$    with $\pol(s) = -$ and $\pol(s') = +$  and $\sig_2(s')  \leq_1\sig_2(s)$   in $\al_1$.
 
 We now describe a test $\tau:T \to A^\perp\vvbar \plmove$ which will discriminate between $\sig_1$ and $\sig_2$.
Let $T'_1$ be the elementary event structure comprising events $T_1\eqdef \sig_1 x_1$ saturated with all accessible Opponent moves (note, in $A^\perp$), \ie~events
$$ 
T_1'= \set{a\in A}{\pol_{A^\perp} ([a]\setdif T_1) \subseteq\setof -}\,
$$
with order
that of $A^\perp$ augmented with  $\sig_2(s')  \leq_1\sig_2(s)$ for every choice $\theta(x_2) = (s,s')$ where $x_2\in M_2$ and $\sig_2 x_2 = \sig_1 x_1$;    the ensuing   
relation on $T_1$ is included in the total order $\al_1$ so forms a partial order in which every element has only finitely many elements below it.  (By design, $T'_1$ ``disagrees'' with the causal dependency of each $x_2\in\conf{S_2}$ for which $\sig_2 x_2 =\sig_1 x_1$.) 
The polarities of events of $T'_1$ are those of its events in $A^\perp$.  On $T'_1$ the map $\tau$ takes an event to its same event in $A^\perp$. 

Let $T$ be the \eswp~obtained from $T_1'$ by adjoining a fresh `success' event $\plmove$ with additional causal dependency so 
$t_1 \leq_T \plmove$  iff $t_1$ is $-$ve
; as noted above there has to be at least one +ve event in $x_1$ and thus, by the reversal of polarity, at least one
$t_1\in T_1$ of $-$ve polarity.  Then the obvious map  $\tau:T \to A^\perp\vvbar \plmove$ is a strategy, and a suitable test for $\sig_1$ and $\sig_2$.  

We have (i) $\sig_1$ may pass $\tau$, while (ii) it is not the case that  $\sig_2$ may pass $\tau$. 

To see (i), remark that the relation of causal dependency on $T_1$  is included in the the total order of the trace $\al_1$ of $x_1$.  Hence  $\tau\sncirc \sig_1$ has a successful configuration $(T_1\cup\setof\plmove) \sncirc x_1$.   

 To show (ii), consider any finite configuration of $\tau\sncirc \sig_2$.  It has the form $w\sncirc x_2$ where $w\in\conf T$   and $x_2\in\conf{S_2}$.  The configuration $w\sncirc x_2$ is unsuccessful because $\plmove\notin w$, as we now show.   
By design,  $\tau$  and $\sig_2$
enforce opposing causal dependencies on a pair of  synchronisations needed for $T_1 \sncirc x_2$ to be defined whenever $x_2\in \conf{S_2}$ with $\sig_2 x_2 =T_1$.  At least two events of opposing polarity in $T_1$ are excluded from any pairing $w\sncirc x_2$; one must be a $-$ve event of $T_1$ on which  $\plmove$ causally  depends;  hence $\plmove\notin w$. 
\end{proof}

That the characterisation of the `may'  preorder 
above only depends on finite traces  is not surprising, and familiar from previous work; the full-abstraction results of Dan Ghica and Andrzej Murawski w.r.t.~`may' behaviour rely only on finite traces~\cite{murawski}.

Clearly the proof above does not rely on stopping configurations or tests being bare rather than pure strategies; the test used in the proof patently has no neutral events. The extra discriminating power of tests based on bare strategies, illustrated in Example~\ref{ex:ptestsmoredescrim}, does play an essential role in the analgous result in the `must' case, to be considered now.  

\subsection{Characterisation of the `must' preorder}

Recall an event structure $E =(E,\leq, \Con)$ is {\em consistent-countable} iff there is a function $\chi:E\to \omega$ from the events   such that 
$$\setof{e_1, e_2}\in\Con \ \&\ \chi(e_1) =\chi(e_2) \implies e_1=e_2\,.$$  
Any configuration $x\in\iconf E$  of a   consistent-countable event structure $E$ is countable and so may be serialised as 
$$
x= \setof{ e_1, e_2, \cdots, e_n, \cdots }
$$
so that $\setof{e_1, \cdots, e_n}\in\conf E$ for any finite subsequence.  
For the must case we assume that games are   consistent-countable.  It follows that strategies $\sig:S\to A$  in consistent-countable games  $A$ have $S$   consistent-countable. W.r.t.~such a strategy $\sig$, we  have traces of all configurations. 

\begin{theorem}\label{thm:mustchar}
Assume game $A$ is   consistent-countable.  Let $(\sigma_1, M_1)$ and  $(\sigma_2, M_2)$ be stopping strategies in $A$. Then, 

$(\sigma_2, M_2)$ must pass $\tau$ implies  $(\sigma_1, M_1)$ must pass $\tau$, 
for all 
 tests $\tau$, 

iff

\noindent
all traces of stopping configurations  $M_1$ are traces of stopping configurations $M_2$.
\end{theorem}
\begin{proof}
{\it ``if'':}
Assume all traces of stopping configurations  $M_1$ are traces of stopping configurations $M_2$.
A stopping configuration of $\tau\sncirc \sig_1$ has the form 
$w\sncirc x_1$ where $w$ and $x_1$ are stopping configurations of $\tau$ and $\sig_1$, respectively. A   serialisation of $w\sncirc x_1$ into a (possibly infinite) sequence induces a serialisation of $x_1\in M_1$.  By assumption, there is $x_2\in M_2$ with the same trace in $A$ as $x_1$.  Consequently, $w\sncirc x_2$ is a configuration of $\tau\sncirc \sig_2$ with the same image in $A\vvbar \plmove$.  Moreover, $w\sncirc x_2$ is a stopping configuration of $\tau\sncirc \sig_2$.   Supposing $(\sigma_2, M_2)$ must pass a test $\tau$, the image of $w\sncirc x_2$ contains $\tick$ whence the image of $w\sncirc x_1$ contains $\tick$ ensuring $(\sigma_1, M_1)$ must pass a test $\tau$.\\

\noindent
{\it ``only if'':}  We show the contraposition: assuming not all  traces of stopping configurations  $M_1$ are traces of stopping configurations $M_2$, we produce a test $\tau$  for which 
 $(\sigma_2, M_2)$ must pass $\tau$ while  it is not the case that $(\sigma_1, M_1)$ must pass $\tau$.  

Assume a trace $\al_1$ of $x_1\in M_1$ is not a trace of any $x_2\in M_2$.  

In particular, consider any $x_2\in M_2$ with $\sig_2 x_2 = \sig_1 x_1$.  Then, any trace of $x_2$ differs from $\al_1$. 
 By Lemma~\ref{lem:disagreement-lem},   there are $s, s'\in x_2$ such that $s\imc_2 s'$  with $\pol(s) = -$ and $\pol(s') = +$ and 
$\sig_2(s')  \leq_1\sig_2(s)$   in the total order $\al_1$.


Thus for each  $x_2\in M_2$ with $\sig_2 x_2 = \sig_1 x_1$ we can choose $\theta(x_2) = (s,s')$ 
so that $s\imc_2 s'$ in $x_2$    with $\pol(s) = -$ and $\pol(s') = +$  and $\sig_2(s')  \leq_1\sig_2(s)$   in $\al_1$.

We build an event structure with polarity $T$ and  a test as bare strategy $\tau:T\to A^\perp\vvbar N \vvbar \plmove$. We build the events of  $T$ as  $T'_1\cup N\cup T_2 
$, a  union of sets of events, assumed disjoint, described as follows.

\begin{itemize}
\item
Let $T'_1$ be the elementary event structure comprising events $T_1\eqdef \sig_1 x_1$ saturated with all accessible Opponent moves, \ie~events
$$ 
T_1'= \set{a\in A}{\pol_{A^\perp} ([a]\setdif T_1) \subseteq\setof -}\,
$$
with order
that of $A$ augmented with  $\sig_2(s')  \leq_1\sig_2(s)$ for every choice $\theta(x_2) = (s,s')$ where $x_2\in M_2$ and $\sig_2 x_2 = \sig_1 x_1$;    the ensuing   
relation on $T_1$ is included in the total order $\al_1$ so forms a partial order in which every element has only finitely many elements below it.  (By design, $T'_1$ ``disagrees'' with the causal dependency of each $x_2\in M_2$ for which $\sig_2 x_2 =\sig_1 x_1$.) 
The polarities of events of $T'_1$ are those of its events in $A^\perp$.  On $T'_1$ the map $\tau$ takes an event to its same event in $A^\perp$. 
\item
$N$ comprises a copy of the set of events  of $-$ve polarity in $T_1$; all the events of $N$ have neutral polarity; an event of $N$ is sent by $\tau$ to its copy.  
\item
$T_2$ comprises a copy of the set of events   
$T_1$; all the events of $T_2$ have +ve polarity; they are all sent by $\tau$ to $\tick\eqdef (3,\plmove)$.
\item
Causal dependency on $T$ is that of $T'_1$ augmented with dependencies from events of $T_1$ of $-$ve polarity 
to their corresponding copies in $N$.
\item
The consistency relation of $T$ is that minimal relation which ensures that: any two distinct events of $T_2$ are in conflict; a +ve event of $T_1$ conflicts with its corresponding copy in $T_2$; and 
a neutral event in $N$ conflicts with its corresponding copy in $T_2$. Formally, 
$$\eqalign{
X&\in \Con_T \hbox{ iff }X\fsubseteq T_1\cup N\cup T_2 
 \ \&\
\mod{X\cap T_2} \leq 1
 \ \&\  \cr
&(\all t_1\in X\cap T_1^+,  t_2\in X\cap T_2. \ t_1, t_2 \hbox{ are not copies of a common event}) \ \&\  \cr
&(\all n \in X\cap N, 
t_2\in X\cap T_2.\ 
  n, t_2 \hbox{ are not copies of a common event}).
}
$$
\end{itemize}  

Note that all the events over $\tick$, which together comprise the set $T_2$, can occur initially but can become blocked as moves are made in $T_1$.  
In particular, 
the set $T_1 \cup N$ is a $+$-maximal configuration of $T$ with image in $A^\perp\vvbar N \vvbar \plmove$ not containing any event over $\tick$. On the other hand any $+$-maximal configuration of $T$ not including all the events $T_1$ will contain an event over $\tick$.    Hence $\St(\tau)$ has an unsuccessful stopping configuration consisting of precisely all the events of $T_1$---it does not have an event over $\tick$---while all  stopping configurations of $\St(\tau)$ which do not contain all the events of  $T_1$ are successful---they contain an event over $\tick$. 
 
Consequently, 
(i) it is not the case that $(\sig_1, M_1)$ must $\tau$, while (ii) $(\sig_2, M_2)$ must $\tau$.  To see (i), remark that the relation of causal dependency on $T_1$  is included in the the total order of the trace $\al_1$ of $x_1$.  Hence  $\St(\tau)\sncirc \sig_1$ has a stopping configuration $T_1\sncirc x_1$ which is unsuccessful and thus  $(\sig_1, M_1)$ fails the must test $\tau$.  To show (ii), consider any stopping configuration of $\St(\tau)\sncirc \sig_2$.  It comprises $w\sncirc x_2$ where $w$ is a stopping configuration of $\St(\tau)$ and $x_2\in M_2$, a stopping configuration of $\sig_2$.  Now 
$w\not\supseteq T_1$,
as by design  $\tau$  and $\sig_2$
enforce opposing causal dependencies on a pair of  synchronisations needed for $T_1\sncirc x_2$ to be defined whenever $x_2\in M_2$ with $\sig_2 x_2 =T_1$.  Thus $w$ is successful in that it contains an event over $\tick$. Hence $(\sig_2, M_2)$ must pass $\tau$. This completes the proof.
 \end{proof}
 
\noindent
{\bf Remark.} 
By Example~\ref{ex:ptestsmoredescrim}, the result above would not hold if tests were based solely on pure strategies.  

\begin{example}{\rm 
Let $A$ be the game
$$
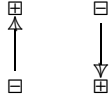
\xymatrix@R=18pt@C=20pt{
\plmove&\opmove\ar@{|>}[d]\\
\ar@{|>}[u]\opmove&\plmove\\
}$$
Let $\sig_1$ be the stopping strategy given by the identity map $\id_A:A\to A$ together with the +-maximal configurations of $A$.  
Let $\sig_2$ be the stopping strategy derived from the event structure
$$
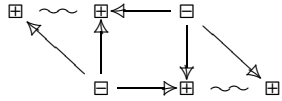
\xymatrix@R=18pt@C=20pt{
\plmove\ar@{~}[r] & \plmove&\ar@{|>}[l]\opmove\ar@{|>}[d]\ar@{|>}[dr]\\
&\ar@{|>}[u]\opmove\ar@{|>}[r]\ar@{|>}[ul]&\plmove& \plmove\ar@{~}[l]\\
}$$
in which there are additional occurrences of Player moves awaiting both moves of Opponent; the map to $A$ is the obvious one and its stopping configurations the +-maximal ones.  It can be checked that the two stopping strategies share the same traces of stopping configurations so are   `must' equivalent.
}\endex\end{example}

 Infinite stopping configurations play an essential role in the `must' behaviour of strategies:

\begin{example}\label{ex:infstopping}{\rm  
 Let the game $A$ consist of an infinite chain of alternating Player-Opponent moves:
$$
\xymatrix@R=2pt@C=20pt{
\plmove\ar@{|>}[r] &\opmove\ar@{|>}[r] &\plmove\ar@{|>}[r] &\opmove\ar@{|>}[r]    &\ \cdots\  \ar@{|>}[r] &
\plmove\ar@{|>}[r] &\opmove\ar@{|>}[r] &\  \cdots}
$$

The strategy $\sig_1$ has as consistent components a copy of $A$ itself and copies of all its initial finite sequences of events ending with an Opponent move;   events of different components are inconsistent with each other---we refrain from drawing the wiggly conflicts.  The map $\sig_1$  is obvious. (The construction is an instance of the sum of strategies---see Section~\ref{sec:sumdecomp}.)
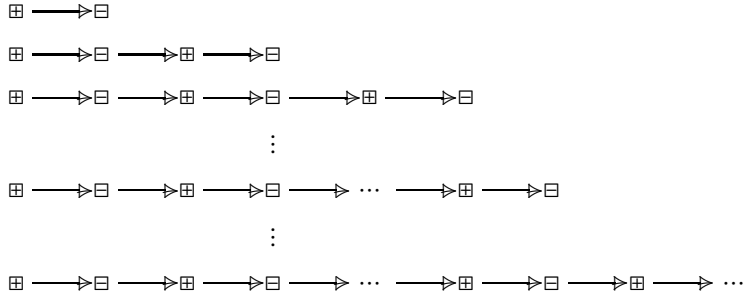
$$
\xymatrix@R=5pt@C=20pt{
\plmove\ar@{|>}[r]&\opmove\\
\plmove\ar@{|>}[r] &\opmove\ar@{|>}[r] &\plmove \ar@{|>}[r] &\opmove\\
\plmove\ar@{|>}[r] &\opmove\ar@{|>}[r] &\plmove\ar@{|>}[r] &\opmove\ar@{|>}[r] &\plmove \ar@{|>}[r] &\opmove\\ 
&&&\vdots&\\
\plmove\ar@{|>}[r] &\opmove\ar@{|>}[r] &\plmove\ar@{|>}[r] &\opmove\ar@{|>}[r]    &\ \cdots\  \ar@{|>}[r] &
\plmove\ar@{|>}[r] &\opmove  &\\
&&&\vdots&\\
\plmove\ar@{|>}[r] &\opmove\ar@{|>}[r] &\plmove\ar@{|>}[r] &\opmove\ar@{|>}[r]    &\ \cdots\  \ar@{|>}[r] &
\plmove\ar@{|>}[r] &\opmove\ar@{|>}[r] &\plmove \ar@{|>}[r] &\ \cdots
}
$$
The lowest component is a copy of $A$.  The stopping configurations are all its +-maximal configurations, \ie~those  sets consisting of  a whole component or an initial subsequence of a component which ends in a Player move. 

The strategy $\sig_2$ is like $\sig_1$ but without the extra infinite component of shape $A$.  Its stopping configurations are all its +-maximal configurations, so are necessarily all finite. 

The 
traces of $\sig_1$'s  and $\sig_2$'s {\em finite} stopping configurations  coincide. However the trace of the infinite stopping configuration of $\sig_1$ cannot be a trace of  $\sig_2$.  
Accordingly, from Theorem~\ref{thm:mustchar}, there is a test strategy which $\sig_2$ must pass while $\sig_1$ does not.  
A distinguishing 
 test $\tau:T\to A^\perp\vvbar\plmove\tick $ has $T$ comprising the infinite event structure shown below.
$$
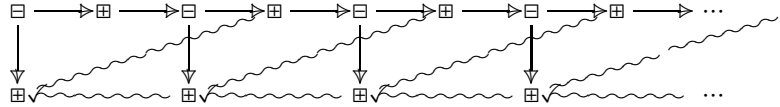
\xymatrix@R=20pt@C=20pt{
\ar@{|>}[d]\opmove\ar@{|>}[r] &\plmove\ar@{|>}[r] &\ar@{|>}[d]\opmove\ar@{|>}[r] &\plmove\ar@{|>}[r]    &\ar@{|>}[d]\opmove \ar@{|>}[r] &
\plmove\ar@{|>}[r] &\ar@{|>}[d]\opmove\ar@{|>}[r] &\plmove\ar@{|>}[r] &\  \cdots&\\
\plmove\tick\ar@{~}[rr] \ar@{~}[urrr]\!\!\!\!\!\!&&\plmove\tick \ar@{~}[rr]  \ar@{~}[urrr]\!\!\!\!\!\!& & \plmove\tick\ar@{~}[rr]  \ar@{~}[urrr]\!\!\!\!\!\! &&
\plmove\tick \ar@{~}[rr] \ar@{~}[urrr]|{{{ }}}\!\!\!\!\!\!
&&\ \cdots \\
}
$$
The  strategy $\sig_2$ must pass $\tau$; all +-maximal  configurations of $\tau\sncirc \sig_2$ are finite and contain a $\tick$-event.
Whereas the  strategy $\sig_1$ can fail to enable a $\tick$-event 
through its extra infinite stopping configuration. 
}\endex\end{example}

\section{The rigid image of a stopping strategy}\label{sec:axioms}

In this section we rely on the material of Section~\ref{sec:rigidimage}, in particular that a strategy $\sig:S \to A$ in a game $A$ has a rigid image
$$\xymatrix{
S\ar[dr]_\sig\ar@{>>}[r]^{f} & S_0\ar[d]^{\sig_1}\\ 
&A\,,}
$$
where $f$ is rigid epi and  the rigid image $\sig_1$ is a total map; the map $\sig_1$ automatically inherits the properties required of a strategy.  (The construction and key properties of rigid image are unaffected by the extra structure of polarity.)  As has been remarked earlier~\cite{ecsym-notes,madeeasy}, rigid-image strategies have the advantage of forming a category rather than a bicategory.  Extended with stopping configurations they can support `may' and `must' behaviour.  

\begin{definition}{\rm 
Let $(\sig, M_S)$ be a stopping strategy.  Let $\sig_1$ be the rigid image of $\sig$ with accompanying 2-cell $f:\sig \Rightarrow \sig_1$ where $f$ is rigid epi.  We define the {\em rigid image} of  $(\sig, M_S)$ to be  $(\sig_1, f M_S)$.
A {\em rigid-image} stopping strategy is one which is its own rigid-image.  
}\end{definition}

As a direct consequence of the last part of Lemma~\ref{lem:2-cellsmaymust}, we are assured the rigid image of a stopping strategy  does not lose any `may' and `must' behaviour. 

\begin{prop}
A stopping strategy   is both  `may' and `must' equivalent to its rigid image.
\end{prop}

As far as `may' and `must' behaviour is concerned it is sensible to regard two stopping strategies  as equivalent if they share a common rigid image.  Rigid-image equivalence transfers to an equivalence between bare strategies: two bare strategies are equivalent if under $\St$ we obtain equivalent stopping strategies. 
W.r.t.~`may' and `must' behaviour we can choose to work in the category of rigid-image stopping strategies.

What axioms hold of stopping configurations?  Such axioms should be preserved by the composition of stopping strategies and rigid image.  They should also be complete in the sense that any stopping strategy which satisfied them is rigid-image equivalent to the stopping strategy of some bare strategy.
We do not presently know a complete set of axioms for stopping configurations.  Candidate axioms on a stopping strategy $\sig:S\to A$ with stopping configurations $M$ are
$$
\eqalign{
&{\rm Axiom\, (i)}\ \  \all x\in\conf S\exists y\in M.\ x\subseteq y\,, \ \hbox{ and }\  \cr
&{\rm Axiom\, (ii)}\  \ \all y\in M, x\in\iconf S.\  x\subseteq y \ \&\ x \hbox{ is +-maximal  in }S \implies x\in M\,.}
$$
The example below shows why we do not assume all +-maximal configurations are stopping. That property is not preserved by taking the rigid image. 

\begin{example}\label{ex:rigim}{\rm 
In forming the rigid image $\sig_1:S_1\to A$ of a strategy $\sig:S\to A$, related by rigid epi 2-cell $f:\sig\Rightarrow \sig_1$, it is possible to have an infinite configuration of $S_1$ which is not in the direct  image under $f$ of any configuration of $S$; in particular it is possible to have a $+$-maximal configuration of $S_1$ which is not a direct image of any +-maximal configuration $S$. For example, let $A$ comprise an infinite chain of Player events.  Take $S$ to be the sum of all finite subchains.  The rigid image of $S$ is $A$ itself which has +-maximal configuration comprising all the events in the infinite chain, not the image of any configuration of $S_1$.  Thus, in forming the rigid image of a stopping strategy, we cannot assume that all the +-maximal configurations of the rigid image are stopping.
}\endex\end{example} 

\section{Strategies as concurrent processes}

The paper~\cite{stratsproc} is a closely related study of concurrent strategies from the perspective of concurrent processes, considering how concurrent games and strategies are objects which we can program. Concurrent strategies are shown to support operations yielding an economic yet rich higher-order concurrent process language, which shares features both with process calculi and nondeterministic dataflow. There a slightly weakened definition of bare strategies plays a key role in providing an operational semantics.  It would be satisfying to complete this story by providing inequational proof systems for `may' and `must' equivalence based on its syntax for strategies, drawing inspiration from the classic work of Hennessy and de Nicola~\cite{hennessy-denicola}.
 
Process calculi often allow unrestricted recursion.  Strategies, as presented here,  form a model of linear logic which restricts the copying of parameters needed in recursive definitions.  
The treatment of unrestricted recursion requires a move to nonlinear strategies  over games with symmetry~\cite{lics14,coHO}.   
The recursive definition of bare strategies and strategies  can follow classical ideas; 2-cells include the rigid embeddings and inclusions of~\cite{icalp82}.  Less clear is how to carry out recursive definitions directly on stopping strategies; the 2-cells we have chosen would seem to be too restrictive; and the nature of stopping configurations, that they can be infinite without finite approximations, would push the development into non-continuous operations---nonstandard, if not in itself a bad thing.   

A treatment of {\em winning} concurrent strategies has been presented~\cite{lics2012}. Informally a strategy is winning if it must end up in a winning configuration of the game regardless of the behaviour of Opponent.  In idea this is very close to  {\em controllability} in~\cite{ATL}.  Because the semantics of composition of composition of strategies in~\cite{lics2012} is inattentive to the possibilities of deadlock and divergence,  a strategy which is obtained as a composition may be deemed winning there and yet possibly deadlock or diverge before reaching a winning configuration~\cite{stratsproc}.  Fortunately the treatment of winning strategies {\it ibid.}~generalises straightforwardly to stopping strategies which keep track of deadlock and divergence, and thus repair this defect.   The role of +-maximal configurations in~\cite{lics2012} is replaced by that of stopping configurations: a bare or stopping strategy is winning iff all its stopping configurations image to  winning strategies in the game.  
 
Forearmed with concurrent strategies and games with symmetry it would be interesting to revisit  old ideas extending  testing to other equivalences beyond those of `may' and `must'~\cite{abramskytesting}.  Certainly one could wish for a better integration of games and strategies with the classical work on concurrency, process algebra and its equivalences included.  The medium of concurrent games and strategies based on event structures also provides an inroad  into the formalisation and analysis of probabilistic and quantum languages and processes~\cite{Probstrats,ProbPCF,quantstrat}.

\section*{\vskip -0.3in 
Acknowledgments}
 Thanks to Jonathan Hayman, Sacha Huriot, Martin Hyland, Marc Lasson, Conor McBride and Marc de Visme for encouragement and helpful discussions.  Thanks to the anonymous referee.  Support of  Advanced Grant ECSYM (2011-17) of the  European Research Council is acknowledged with gratitude.  The second author gratefully acknowledges support by ANR project DyVerSe
(ANR-19-CE48-0010-01) and Labex MiLyon (ANR-10-LABX-0070) of Universit\'e
de Lyon, within the program ``Investissements d’Avenir''
(ANR-11-IDEX-0007), operated by the French National Research Agency (ANR).

 \bibliographystyle{splncs} 
\bibliography{biblio}

\appendix

 \end{document}